\documentclass{article}

\usepackage[english]{babel}
\usepackage[utf8x]{inputenc}
\usepackage[T1]{fontenc}

\usepackage{ dsfont }
\usepackage{placeins}
\usepackage{multirow}
\usepackage{tabularx} 
\usepackage{lscape}
\usepackage{algpseudocode}
\usepackage{color,soul}
\usepackage{hyperref} 
\usepackage{csvsimple}
\usepackage{longtable}
\usepackage{arydshln,amssymb,color}
\usepackage{multicol}
\usepackage{bm}
\usepackage{tikz}
\usetikzlibrary{arrows.meta}
\usepackage{amsmath}
\usepackage{graphicx}
\usepackage[colorinlistoftodos]{todonotes}
\usepackage{booktabs} 
\usepackage{xtab,booktabs}
\usepackage{footnote}
\makesavenoteenv{tabular}
\makesavenoteenv{table}
\usepackage{ dsfont }
\usepackage{ stmaryrd }
\usepackage{flushend}
\usepackage{amsthm}
\usepackage{fp}
\usepackage{xfp}
\usepackage{siunitx}
\usepackage{listings}
\usepackage{subcaption}

\usepackage[linesnumbered,algoruled,boxed,lined,noend]{algorithm2e}
\SetKwInOut{KwParam}{Parameters}
\DeclareMathOperator*{\abs}{abs}

\usepackage{csvsimple,longtable,booktabs}
\usepackage{breqn}

\newtheorem{mydef}{Definition}

\newtheorem{mytheorem}{Theorem}
\newtheorem{mylemma}{Lemma}
\newtheorem{myassumption}{Assumption}

\newcommand*{\lwhere}{\ \ | \ \ }

\newcommand\nmachines{8}

\newcommand\aparamregression{-0.62280}
\newcommand\bparamregression{2008}

\newcommand\minusbdivaparamregression{3223.49}

\newcommand\singlethreadscoreofCeleronNFourOneZeroZero{1012}
\newcommand\singlethreadscoreofXeonEFiveFourFourZero{1219}

\begin{document}

\title{On the Fair Comparison of Optimization Algorithms in Different Machines}

	\author{
		Etor Arza, BCAM - Basque Center for Applied Mathematics\\
		Josu Ceberio, University of the Basque Country UPV/EHU\\
		Ekhiñe Irurozki, Télécom Paris\\
		Aritz Pérez, BCAM - Basque Center for Applied Mathematics
		}
	
	\maketitle


	\begin{abstract}
		An experimental comparison of two or more optimization algorithms requires the same computational resources to be assigned to each algorithm.
		When a maximum runtime is set as the stopping criterion, all algorithms need to be executed in the same machine if they are to use the same resources.
		Unfortunately, the implementation code of the algorithms is not always available, which means that running the algorithms to be compared in the same machine is not always possible.
		And even if they are available, some optimization algorithms might be costly to run, such as training large neural-networks in the cloud.
		
		In this paper, we consider the following problem: how do we compare the performance of a new optimization algorithm B with a known algorithm A in the literature if we only have the results (the objective values) and the runtime in each instance of algorithm A?
		Particularly, we present a methodology that enables a statistical analysis of the performance of algorithms executed in different machines.
		The proposed methodology has two parts.
		First, we propose a model that, given the runtime of an algorithm in a machine, estimates the runtime of the same algorithm in another machine.
		This model can be adjusted so that the probability of estimating a runtime longer than what it should be is arbitrarily low.
		Second, we introduce an adaptation of the one-sided sign test that uses a modified \textit{p}-value and takes into account that probability.
		Such adaptation avoids increasing the probability of type I error associated with executing algorithms A and B in different machines.
	\end{abstract}

	\section{Introduction}

	Finding an appropriate reference point for evaluating the performance of an optimization algorithm is not trivial.
	The key question is: when can we say that the performance of an optimization algorithm is good?
	The answer depends on how we define good performance.
	A possible solution is to compare the performance of several algorithms in the same problem.
	This comparison can show that some algorithms perform better than others on average. 
	Precisely, this is what a comparative study among different algorithms tries to accomplish: analyze the relative performance of a set of algorithms.

	However, a different use of computational resources can lead to observing false differences between the performance of the algorithms.
	Appendix~\ref{appendix:extra_execution_time_typeI_error} presents an illustrative example.
	We say that two algorithms use the same amount of computational resources when they both take the same time to complete in the same machine.
	Usually, this is achieved by executing both algorithms on the same machine and setting a common maximum runtime as stopping criterion.

	Over the last few decades, the computational capabilities of computers have significantly increased~\cite{nordhaus_2007}.
	This means that code executed a decade ago is expected to run faster on current hardware.
	Additionally, in some fields of optimization such as natural language processing, good performing models like the GTP-3~\cite{brownLanguageModelsAre2020} are currently being trained with a lot of computation power.
	Unfortunately, training these types of models is often economically unviable for most researchers~\cite{sharir2020cost}.
	Furthermore, as stated by Hutson~\cite{doi:10.1126/science.359.6377.725}, many works do not include the code to reproduce the experiments.
	These and other issues limit the reproducibility of the comparison of algorithms, as it is not always possible to execute all the algorithms being compared in the same machine.

	One way to overcome this limitation is to adjust the runtime of one of the algorithms being compared such that the algorithm executed on the slower machine is compensated with extra computation time.
	Let us now look at a practical example.
	Let us imagine that a researcher reads a paper in which algorithm $A$ is executed in machine $M_1$, taking time $t_1$.
	Now, the researcher wants to compare a new algorithm, $B$, with $A$, but has no access to algorithm $A$ nor to machine $M_1$.
	Instead, the researcher only has access to machine $M_2$ and algorithm $B$.
	In this case, he/she can execute algorithm $B$ in machine $M_2$ for time $t_2$.
	The runtime $t_2$ needs to be set in such a way that both algorithms are given the same computational resources.
	To achieve this, $t_2$ needs to be equal to the equivalent runtime: the time it takes to replicate in machine $M_2$ the exact optimization process that was carried out in machine $M_1$ (with algorithm $A$).

	The exact equivalent runtime $t_2$ can be obtained if the exact optimization process that was carried out in machine $M_1$ is replicated in machine $M_2$.
	This implies executing algorithm $A$ in machine $M_2$, which defeats the purpose of using an equivalent runtime.
	Fortunately, an estimation of the equivalent runtime $\hat{t}_2$ can be used instead.
	The estimation is carried out taking into account the computational capabilities of machines $M_1$ and $M_2$, denoted as $s_1$ and $s_2$ in the rest of the paper.
	Table~\ref{table:summary_terms_paper} offers a brief overview of the terms used in the paper.
	
	\begin{table}[]
		\centering
		\begin{footnotesize}
			\caption*{\small \textbf{A summary of the notation of the paper}}
			\begin{tabular}{l|c|l}
				Name                                                                                    &      \scriptsize Notation      & Explanation                                                   \\ \hline
				\multirow{3}{3.4cm}{Optimization algorithm $A$}                                         &     \multirow{3}{1em}{$A$}     & The optimization algorithm that \underline{is not executed} in \\
				                                                                                        &                                & the comparison. Instead, already published results             \\
				                                                                                        &                                & of this algorithm are used in the comparison.                  \\ \hline
				\multirow{2}{3.5cm}{Optimization algorithm $B$}                                         &     \multirow{2}{1em}{$B$}     & The optimization algorithm that \underline{is executed} to     \\
				                                                                                        &                                & obtain the results to be compared.                             \\ \hline
				\multirow{2}{3.3cm}{Machine $M_1$}                                                      &    \multirow{2}{1em}{$M_1$}    & The machine in which algorithm $A$ was executed.               \\
				                                                                                        &                                & We have \underline{no access} to this machine.                 \\ \hline
				\multirow{2}{3.3cm}{Machine $M_2$}                                                      &    \multirow{2}{1em}{$M_2$}    & The machine in which algorithm $B$ is executed to              \\
				                                                                                        &                                & obtain the results used in the comparison.                     \\ \hline
				\multirow{2}{3.3cm}{The score of machine $M_1$}                                         &    \multirow{2}{1em}{$s_1$}    & A measure proportional to the computational                    \\
				                                                                                        &                                & capability  of machine $M_1$.                                  \\ \hline
				\multirow{2}{3.3cm}{The score of machine $M_2$}                                         &    \multirow{2}{1em}{$s_2$}    & A measure proportional to the computational                    \\
				                                                                                        &                                & capability  of machine $M_2$.                                  \\ \hline
				\multirow{3}{3.3cm}{The runtime of $A$ in machine $M_1$}                                &    \multirow{3}{1em}{$t_1$}    & The stopping criterion (in terms of maximum                    \\
				                                                                                        &                                & runtime) that was used in the execution of                     \\
				                                                                                        &                                & algorithm $A$ in machine $M_1$.                                \\ \hline
				\multirow{3}{3.3cm}{The equivalent runtime of $A$ in machine $M_2$}                     &    \multirow{3}{1em}{$t_2$}    & The time it takes to replicate in machine $M_2$ the            \\
				                                                                                        &                                & exact optimization process (with algorithm $A$)                \\
				                                                                                        &                                & that took time $t_1$ in machine $M_1$.                         \\ \hline
				\multirow{3}{3.3cm}{The estimated equivalent runtime of algorithm $A$ in machine $M_2$} & \multirow{3}{1em}{$\hat{t}_2$} & An estimation of the equivalent runtime $t_2$.                 \\
				                                                                                        &                                & This value is used as the stopping criterion of                    \\
				                                                                                        &                                & algorithm $B$, which is executed in machine $M_2$.             \\ \hline
			\end{tabular}
		\end{footnotesize}
		\caption
		{
			A summary of the notation and terms considered in this paper.
		}
		\label{table:summary_terms_paper}
		
	\end{table}

	\textbf{Related work:} 
	
	Dominguez et al.~\cite{dominguez_methodology_2012} proposed adjusting the runtimes of the algorithms by assigning a shorter CPU runtime to the algorithms executed in machines with faster CPUs, thus making the execution of algorithms in different machines comparable.
	To estimate the CPU capabilities of each machine, they proposed using the \textit{dhrystone2} score~\cite{weicker_dhrystone_1988}.

	The methodology introduced by Dominguez et al. is limited in three ways.
	Firstly, there is neither theoretical nor experimental justification for the prediction model of the equivalent runtime.
	Secondly, their prediction model is fixed and cannot be adjusted control the probability of type I error.
	And thirdly, their methodology does not take into account the probability of predicting an equivalent runtime longer or shorter than the true equivalent runtime, and thus, may introduce undesired biases to the comparison of the performance of algorithms.
	Without a corrected statistical model, it is not possible to take into account this probability and control the probability of type I error.

	An increase in the probability of type I error when deciding if algorithm $B$ is better than $A$ can be problematic.
	In the context of performance comparison of optimization algorithms (with null-hypothesis statistical tests), a type I error is defined as finding a statistically significant difference in the performance of the algorithms, when in reality, there is none.
	Making a type II error, on the other hand, means not finding a statistically significant difference in the performance of the algorithms, when in reality, the performances are different.
	In this context, making a type II error is preferred to a type I error: falsely concluding a nonexistent difference in the performance of the algorithms is worse than not finding a statistically significant difference in the performance of the algorithms.

	In fact, failing to reject the null hypothesis does not imply evidence in favor of the null hypothesis. 
	Instead, it only shows a lack of evidence against it.
	In the context of algorithm benchmarking, failing to reject the null hypothesis does not mean that the performance of the algorithms is the same.
	When the null hypothesis is not rejected, the correct conclusion is that there is not enough evidence to show a statistically significant difference between the performance of the algorithms. 
	Therefore, a type II error just means that additional experimentation is needed to verify an existing difference in the performance of the algorithms, which is not an erroneous conclusion in itself.

	\textbf{Proposed methodology:}
	In this paper, inspired by the work of Dominguez et al.~\cite{dominguez_methodology_2012}, we propose a methodology to statistically assess the difference in the performance of optimization algorithms executed in different machines.
	Specifically, the proposed methodology can be used to show that an algorithm $B$ performs statistically significantly better than another algorithm $A$, without executing $A$ and, instead, using the available results of $A$ in terms of the objective function value and the runtime in each instance.
	To that end, we propose a conservative methodology in which the probability of giving algorithm $B$ an unfairly longer time is kept in check by i) proposing a two-parameter estimation\footnote{The estimation proposed in this paper is intentionally conservative and tends to estimate shorter than equivalent runtimes. With this we are able to keep the probability of type I error in check, while increasing the probability of type II error. In the context of algorithm benchmarking, a type I error is worse than a type II error.} of the equivalent runtime with an arbitrarily low probability of estimating an unfairly longer runtime and ii) by modifying the one-sided sign test~\cite{conover1980practical} so that it takes this probability into account.

	Alongside this paper we present a tutorial on how to apply the proposed methodology.
	This tutorial and the code of all the experimentation is available in our \href{https://github.com/EtorArza/RTDHW}{GitHub repository}\footnote{Repository available in  \href{https://github.com/EtorArza/RTDHW}{https://github.com/EtorArza/RTDHW}.}.
	Besides, we also give two examples of how the methodology is applied in this paper.
	It is noteworthy that applying the proposed methodology to compare algorithms in different machines does not involve executing any additional code.

	The rest of the paper is organized as follows: 
	The next section describes and motivates a two-parameter model proposed to estimate the equivalent time.
	Section~\ref{section:modifying_sign_test} presents the modifications made to the sign test to overcome the limitations introduced by the execution of the algorithms in different machines.
	Afterward, in Section~\ref{section:case_study} we introduce two examples in which we apply the proposed methodology.
	Finally, Section~\ref{section:conclusion} concludes the paper and proposes some research lines for future investigation.

	\section{The estimation model of the equivalent runtime}
	\label{section:norm_exec_time}


	Given i) an optimization algorithm, ii) a machine, iii) a problem instance, iv) a stopping criterion and v) a random seed number, executing the optimization algorithm will produce a specific sequence of computational instructions.
	This sequence is completely determined by these five parameters.
	We call this sequence of instructions that is reproducible in any machine the \textit{optimization process}.
	By recording the optimization process carried out with these parameters, we can later reproduce the exact optimization process in another machine.
	Notice that reproducing the optimization process will take a different time in each machine, even though the final result is the same (because the executed sequence of instructions is the same).
	We say that the times required to replicate the same optimization process in different machines are equivalent.

	\vspace{0.15cm}
	\begin{mydef} (Optimization process) \\
	Let $M$ be a machine, $A$ an optimization algorithm, $i$ a problem instance, $t_1$ a stopping criterion, and $r$ a positive integer (the seed for the random number generator).
	We define the optimization process $\rho(M,A,i,t_1,r)$ as the sequence of computational instructions carried out when optimizing instance $i$ with algorithm $A$ and seed $r$ in machine $M_1$ with stopping criterion $t_1$.
	\end{mydef}

	The aim is to compare algorithm $A$ executed in machine $M_1$, with algorithm $B$, executed in machine $M_2$.
	A fair comparison can be carried out by estimating the time it takes to replicate $\rho(M_1,A,i,t_1,r)$ in machine $M_2$ and using the estimated value as the stopping criterion for algorithm $B$ in machine $M_2$. 
	We will sometimes denote the optimization process $\rho(M_1,A,i,t_1,r)$ as $\rho$ for the sake of brevity.
	
	\vspace{0.15cm}
	\begin{mydef} 
	\label{definition:runtime_of_an_optimization_process}
	(Runtime of an optimization process) \\ 	
	Let $M$ be a machine and $\rho$ an optimization process.
	We define the runtime of $\rho$ in $M$, denoted as $t(M,\rho)$, as the time it takes to carry out the optimization process $\rho$ in machine $M$.
	\end{mydef}

	Considering the above definitions, it follows that, $t(M_1,\rho) = t_1$.

	\vspace{0.15cm}
	\begin{mydef} 
	\label{definition:equivalent_runtime}
	(Equivalent runtime) \\
	Let $M_1,M_2$ be two machines, $\rho$ an optimization process and $t(M_1,\rho)$ and $t(M_2,\rho)$ the times required to run $\rho$ in $M_1$ and $M_2$ respectively.
	Then, we say that $t(M_2,\rho)$ is the equivalent runtime of $t(M_1,\rho)$ for machine $M_2$.
	\end{mydef}

	From here on, we will denote $t(M_2,\rho)$, the equivalent runtime of $t_1 = t(M_1,\rho)$ in machine $M_2$, as $t_2$.
	Given $t_1$ (the runtime of optimization process $\rho$ in a machine $M_1$), in the following, we will propose a model to estimate $t_2$ (the equivalent runtime in another machine $M_2$).

	\vspace{0.15cm}
	\begin{myassumption} 
	\label{assumption:k_constant_ratio}
	(Constant ratio of the runtime of two optimization processes) \\
	Let $\rho, \rho'$ be two optimization processes.
	Then, we assume that:
	\begin{equation*}
	\frac{t(M_2,\rho)}{t(M_2,\rho')} \approx \frac{t(M_1,\rho)}{t(M_1,\rho')} 
	\end{equation*}
	for any two machines  $M_1$ and $M_2$.
	\end{myassumption}
	\vspace{0.15cm}

	We assume that the ratio of the runtime of two different optimization processes is constant with respect to the machine in which it is measured.
	In Appendix~\ref{appendix:constant_ratio_tasks}, we justify why this assumption is reasonable.
	This assumption is critical to the estimation model that will later be proposed.
	By using this assumption in the model, a prediction error is introduced.
	Therefore, we will later propose a correction to the model to control this prediction error.

	Based on this assumption, we propose a model to estimate the equivalent runtime of an optimization process in a machine, given its runtime in another machine, as well as the scores (relative to the computational capabilities) of both machines.
	Notice that in Assumption~\ref{assumption:k_constant_ratio}, we use a reference optimization process $\rho'$ to estimate the equivalent runtime of the optimization process $\rho$.
	Any optimization process $\rho'$ can be used as a reference.
	In the following, we will define an optimization process $\rho'$ whose runtime we will be able to estimate with the scores $s_1$ and $s_2$ of the machines.
	This will allow the estimation of the equivalent runtime $t_2$ without executing any reference optimization processes, as shown in Figure~\ref{fig:estimation_diagram}.
	
	\begin{figure}
		\centering
		\caption*{$\hspace{1.9em}$Diagram of the estimation of the equivalent runtime}
		\includegraphics[width=0.7\linewidth]{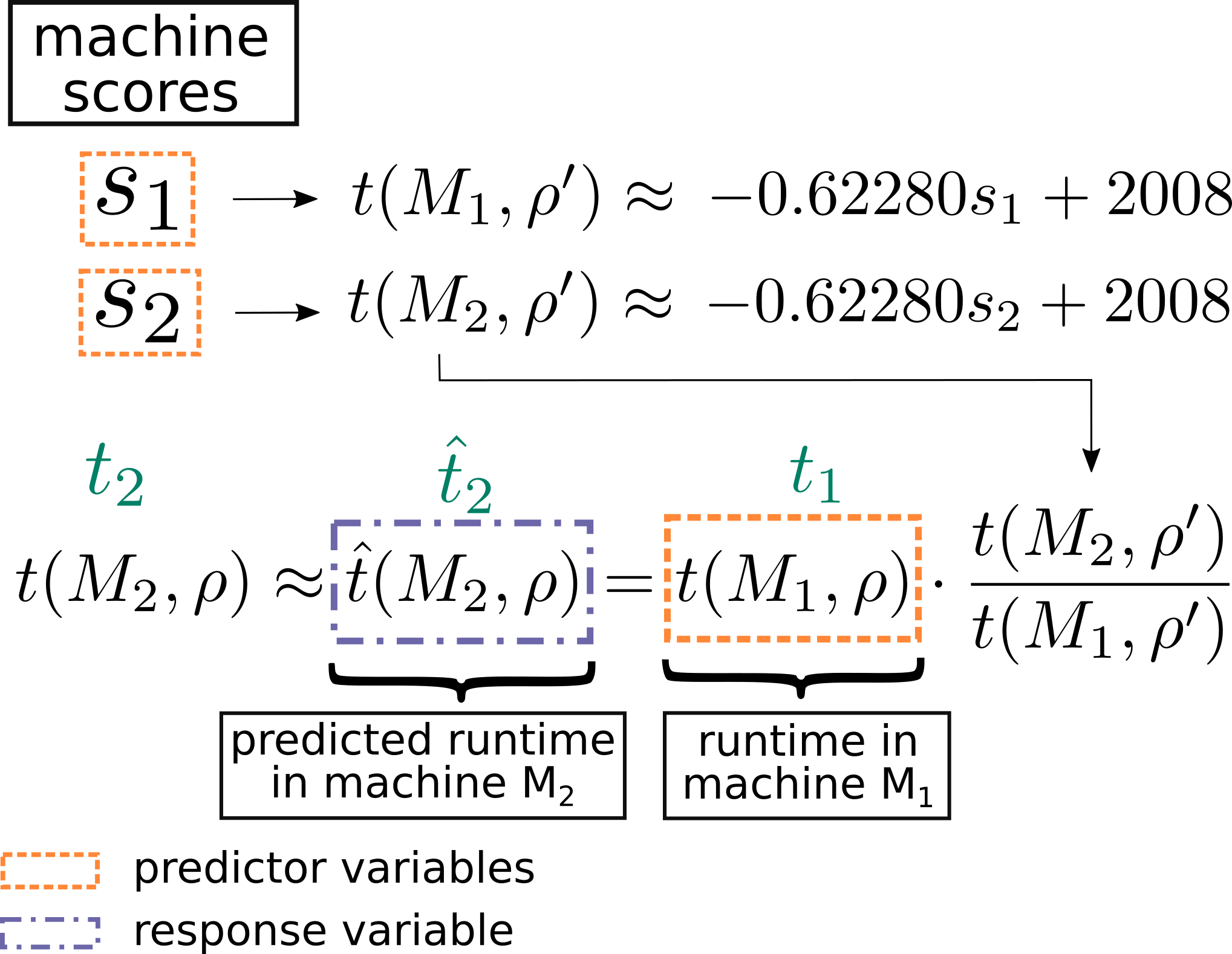}
		\vspace{0.4cm}
		\caption{
			Estimated equivalent runtime of $\rho$ in machine $M_2$ (the response variable $\hat{t}(M_2,\rho)$).
			The estimation is carried out with three predictor variables: the machine scores $s_1$ and $s_2$ and $t(M_1,\rho)$.
		}
		\label{fig:estimation_diagram}
	\end{figure}

	Let us now define the optimization process $\rho'$, whose runtime can be estimated.
	Recall that an optimization process is just a sequence of computational instructions that can be reproduced in any machine.
	Aiming to obtain a more diverse sequence of computational instructions, we define the optimization process $\rho'$ as the computational instructions generated by consecutively executing four different optimization algorithms in 16 problem instances.
	Each of the 64 executions involves solving a permutation problem with an optimization algorithm, with a stopping criterion of a maximum of $2\cdot10^6$ evaluations (see Appendix~\ref{appendix:experimentation_execution} for details on the optimization problems and algorithms used).

	\begin{figure}
		\centering
		\caption*{PassMark single-thread score and the runtime $\rho'$}
		\includegraphics[width=0.8\linewidth]{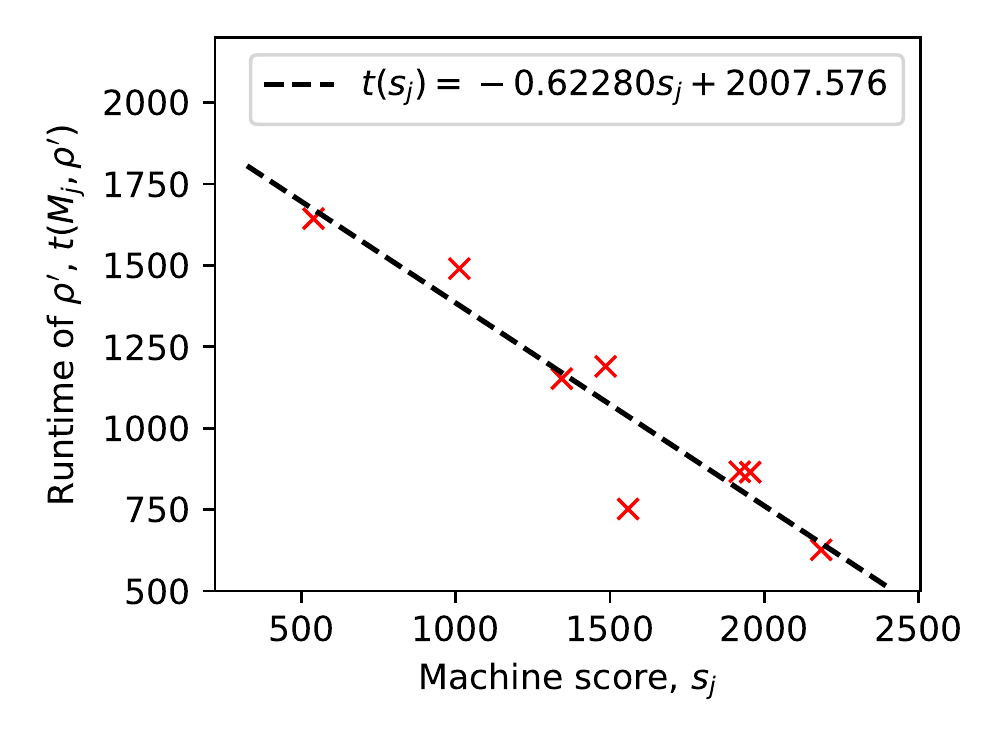}
		\caption{
			The runtime of $\rho'$, the optimization process used as a reference to define the regression model.
			Every point represents a different CPU, each with a different machine score and runtime of $\rho'$ in this machine.
		}
		\label{fig:passmark_base_algorithm_regression}
	\end{figure}

	The runtime of the optimization process $\rho'$ in a machine can be estimated with its machine score. 
	In this paper, we measure the score of a machine (its computational capability) in terms of its \href{https://github.com/EtorArza/RTDHW/blob/master/cpu_scores.md}{PassMark} single-thread CPU score
	\footnote{
		The PassMark CPU score is one of the most popular CPU benchmark scores, with over 3500 CPUs listed on their website.
		In this paper, we use the single-thread score of this benchmark.
		A higher value of the score is associated with a better relative performance of the CPU.
		The PassMark scores change over time, as new CPUs and consumer demand for computation evolves.
		The scores considered in this paper can be looked up in the file    \href{https://github.com/EtorArza/RTDHW/blob/master/cpu_scores.md}{cpu\_scores.md} in our GitHub repo.
	},
	although adapting the proposed methodology to other benchmarks is also possible.
	The advantage of using the PassMark score is that the PassMark website has a large collection of CPUs with their scores.
	In Figure~\ref{fig:passmark_base_algorithm_regression}, we show the machine score and runtime of $\rho'$ for each of the $\nmachines$ machines considered in this paper (the list of machines is available in Appendix~\ref{appendix:experimentation_execution}).
	We see that the relationship between the runtime of $\rho'$ and the machine score is linear.
	Every point in Figure~\ref{fig:passmark_base_algorithm_regression} represents a different machine. 

	Based on the figure shown, we infer that i) linear regression is suitable to model the runtime of $\rho'$ with respect to the machine score, and ii) the machine score has a good correlation with the runtime of $\rho'$.
	Note that using a linear estimation of $\rho'$ has some limitations in its applicability that will be addressed in Section~\ref{section:limitations}.
	With this in mind, we define the estimation of the runtime in a machine:
	
	\vspace{0.15cm}
	\begin{mydef} 
		\label{definition:linear_regression_passmark}
		(Prediction of the runtime of $\rho'$ in a machine)\\
		Let $M_j$ be a machine and $s_j$ its machine score. Then, the runtime of $\rho'$ in  $M_j$ is modeled as
		\begin{equation*}
		t(M_j,\rho') \approx  \aparamregression s_j + \bparamregression
		\end{equation*}
		where $s_j$ is the score of machine $M_j$.
	\end{mydef}

	Considering together Assumption~\ref{assumption:k_constant_ratio} and Definition~\ref{definition:linear_regression_passmark}, the equivalent runtime can be estimated as:
	
	\begin{equation}
		\label{equation:centered_estimator_runtime}
		t_2 \approx  \hat{t}_2 = t_1 \cdot \dfrac{\aparamregression s_2 +\bparamregression}{\aparamregression s_1 + \bparamregression} = t_1 \cdot \dfrac{\minusbdivaparamregression - s_2}{\minusbdivaparamregression - s_1}
	\end{equation}
	
	where $s_1$ and $s_2$ are the PassMark single-thread scores of the CPUs on machines $M_1,M_2$, respectively.
	Due to the approximation errors in Assumption~\ref{assumption:k_constant_ratio} and Definition~\ref{definition:linear_regression_passmark}, the estimated equivalent runtime $\hat{t}_2 = \hat{t}(M_2,\rho)$ will often differ from the true equivalent runtime value $t_2  = t(M_2,\rho)$.
	This means that when using the estimated equivalent runtime as the stopping criterion for algorithm $B$, sometimes, the runtime will be longer or shorter than the runtime used by algorithm A.
	
	\subsection{Controlling the probability of predicting a runtime longer than the true equivalent runtime}
	
	To statistically assess the uncertainty associated to the comparison of the performance of algorithms A and B, in this methodology, we propose using a one-sided statistical test.
	Under this test, the alternative hypothesis states that the performance of algorithm B is better than the performance of algorithm A.
	As a result, a type I error (erroneously finding a statistically significant difference in the performance of A and B) can only be made when algorithm B performs better than A.

	When a shorter runtime is estimated, algorithm $B$ has an ``unfairly'' shorter stopping criterion for the optimization.
	This implies that the measured performance of $B$ will be worse than or equal to the performance that would have been measured if the actual equivalent runtime were used.
	Consequently, taking into account the one-sided nature of the test, estimating a lower than actual runtime will not increase the probability of type I error (estimating a lower than actual runtime can never help algorithm $B$ perform better than algorithm $A$).
	It might, however, increase the probability of type II error.

%

	As mentioned in the Introduction, making a type II error is better than making a type I error when comparing algorithm performance. 
	This is because, in this context, it is better to miss evidence that can adequately discriminate between two algorithms than to observe a false difference. 
	For example, let us assume that someone publishes algorithm A with a certain performance.
	Now let us assume that a researcher proposes an algorithm B that is a variation of algorithm A.
	If a type II error is made, then B is actually better than A, but the researcher is not able to find enough evidence to support this, which is not in itself an erroneous conclusion.
	However, in a type I error, algorithm B is actually worse or equal to A but the researcher incorrectly identified algorithm B as the better algorithm, and this can be more detrimental to scientific progress.

	To avoid drawing erroneous conclusions, we present a modification to Equation~\eqref{equation:centered_estimator_runtime} so that the probability of estimating a longer time than the actual equivalent runtime stays under a percentage chosen by the user.
	We reformulate the unbiased estimator shown in Equation~\eqref{equation:centered_estimator_runtime} to reduce the probability of estimating a runtime longer than the true equivalent runtime.
	The new biased estimator is defined by multiplying the unbiased estimator with a correction parameter $\gamma\in(0,1]$:

	\vspace{0.15cm}
	\begin{mydef} (Estimation of the equivalent runtime in machine $M_2$)\\
	\label{definition:estimation_execution_time}
	Let $\rho$ be an optimization process, $M_1,M_2$ two machines and $t_1$ the runtime of $\rho$ in machine $M_1$.
	We compute $\hat{t}_2$, the estimate equivalent runtime of $\rho$ for machine $M_2$ as:
	
	\begin{equation*}
	t_2 = t(M_2,\rho) \approx  \hat{t}_2 = t_1 \cdot \dfrac{\minusbdivaparamregression - s_2}{\minusbdivaparamregression - s_1} \cdot \gamma
	\end{equation*}
	
	\end{mydef}
	
	By adjusting $\gamma$, the probability of estimating a longer runtime than the equivalent runtime, $\mathcal{P}[\hat{t}_2 > t_2]$, can be reduced.
	However, adjusting $\gamma$ implies that on average, a shorter runtime is predicted.
	With $\gamma=1.0$, the original, unbiased estimator is obtained.
	A lower value of $\gamma$ is associated to a lower probability of estimating a runtime longer than the true runtime.
	Specifically, the parameter $\gamma$ is equal to $\mathbb{E}[\frac{\hat{t}_2}{t_2}]$: how much shorter the estimated equivalent runtime is than the true equivalent runtime on average.
	For example, when $\gamma = 0.5$, then the equivalent runtime used will be half of the equivalent runtime predicted by the unbiased estimator.

	We estimated the relationship between $\gamma$ and $\mathcal{P}[\hat{t}_2 > t_2]$ and we show the result in Figure~\ref{figure:correctioncoefficienttradeoff}.
	We computed this probability empirically in a cross-validation setting.
	The exact procedure carried out to generate this figure is available in the file \href{https://github.com/EtorArza/RTDHW/blob/master/processing_results/show_linear.py}{show\_linear.py} in our GitHub \href{https://github.com/EtorArza/RTDHW}{repository}.
	In the following, we give additional details on the process carried out to compute the relationship between $\gamma$ and $\mathcal{P}[\hat{t}_2 > t_2]$ shown in the figure.

	The pseudocode of the following process is shown in Algorithm~\ref{algo:compute_prob_forall_gamma}.	
	Given a value of $\gamma$, we start by iterating over all the optimization processes and every possible pair of CPUs (Lines~3-5) and we leave them out of the training data (Lines~6-7).
	Then, we fit the linear regression in Definition~\ref{definition:linear_regression_passmark} and Equation~\eqref{equation:centered_estimator_runtime}, but only with the CPUs and optimization processes in the training data (Line~8).
	The runtime of the left out optimization process in the machine with cpu1 is $t_1$ (Line~11).
	Now, we predict the equivalent runtime of the optimization process left out in the machine with cpu2 with the formula $\hat{t_2} \gets  t_1 \cdot \frac{\alpha - s_2}{\alpha - s_1} \cdot \gamma$, where $\alpha$ was fitted with the training data (Line~12).
	Finally, we measure the proportion of times in which the predicted equivalent runtime for the machine with cpu2, $\hat{t_2}$, was longer than the runtime of the optimization process in that machine, $t_2$.

	\begin{algorithm}[t]
		\begin{footnotesize}
			\DontPrintSemicolon 
			\caption{Compute the probability of estimating a runtime longer than the true equivalent runtime given $\gamma$ in a cross-validation setting}
			\label{algo:compute_prob_forall_gamma}
			
			\KwIn{\ \\
				cpu\_list: The list of all the CPUs used to fit the linear regression.\\
				\noindent process\_list: The list of all the optimization processes.\\
				\noindent $\gamma$: The correction coefficient.\\
			}
			\KwOut{\ \\
				$\mathcal{P}[\hat{t}_2 > t_2]$: The probability of estimating a runtime longer than the true equivalent runtime for the given $\gamma$.
			}
			
			\SetAlgoLined
			\vspace{0.25cm}

			total $\gets 0$\;
			longer\_runtime\_predicted $\gets 0$\;

			\ForAll{\normalfont test\_process \textbf{in} process\_list}{
				\ForAll{\normalfont cpu1 \textbf{in} cpu\_list}{
					\ForAll{\normalfont cpu2 \textbf{in} cpu\_list $\setminus$ \{cpu1\}}{
						train\_cpus $\gets$  cpu\_list $\setminus$ \{cpu1, cpu2\}\;
						train\_processes $\gets$ process\_list $\setminus$ \{test\_process\} \;
						fit the linear regression in Definition~\ref{definition:linear_regression_passmark} and Equation~\eqref{equation:centered_estimator_runtime} with the cpu scores of train\_cpus and the runtimes of the optimization processes in train\_processes when executed in train\_cpus.\;
						$s_1 \gets$ cpu score of cpu1\; 
						$s_2 \gets$ cpu score of cpu2\;
						$t_1 \gets$ runtime of test\_process in the machine with cpu1\;
						$\hat{t_2} \gets  t_1 \cdot \frac{\alpha - s_2}{\alpha - s_1} \cdot \gamma$, where $\alpha$ was adjusted in Line~8.\;
						total $\gets$ total $+ 1$\;						
						\If{$\hat{t_2} > t_2$}{						longer\_runtime\_predicted $\gets$ longer\_runtime\_predicted $+ 1$\;						}

					}
				}			
			}

			\textbf{return} $\dfrac{ \text{longer\_runtime\_predicted}}{\text{total}}$ \;
			
		\end{footnotesize}    
	\end{algorithm}

	\begin{figure}
		\centering
		\caption*{$\hspace{1em}$Estimated runtime and the correction parameter $\gamma$}
		\includegraphics[width=0.7\linewidth]{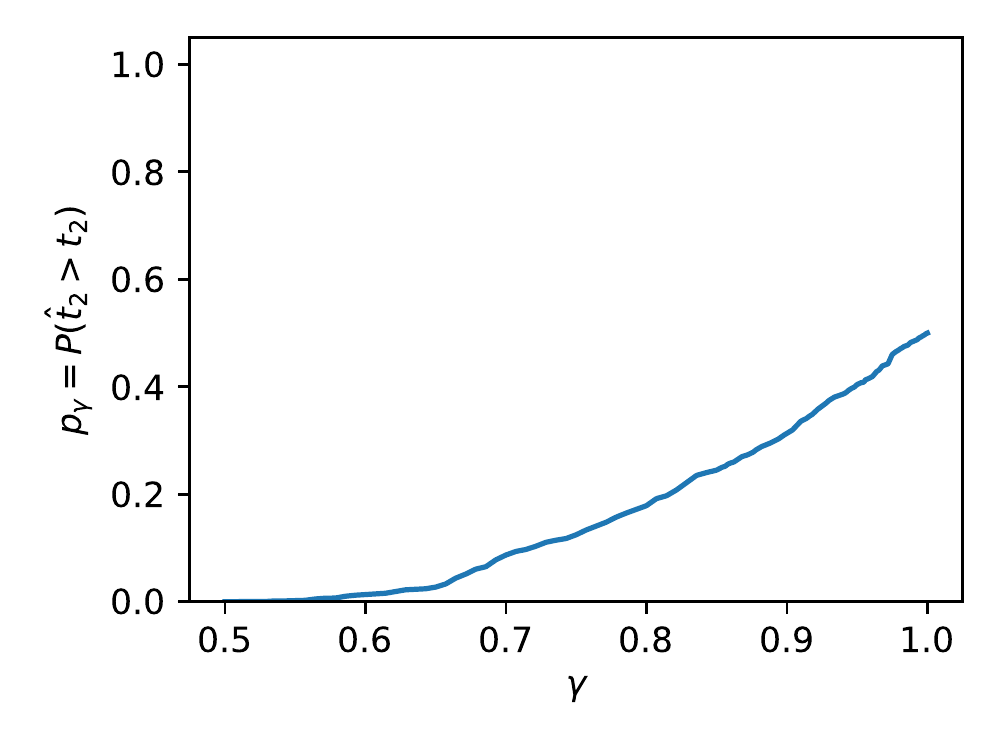}
		\caption
		{
			The probability of estimating a runtime longer than the true equivalent runtime with respect to $\gamma$.
		}
		\label{figure:correctioncoefficienttradeoff}
	\end{figure}

	Instead of considering the parameter $\gamma$, we can also think of the parameter $p_{\gamma} = \mathcal{P}[\hat{t}_2 > t_2]$.
	$p_\gamma$ is the probability of estimating a longer than equivalent runtime associated to $\gamma$.
	It is probably more useful for the user to think of $p_{\gamma}$, because this is what is directly related to the increase of probability of making a type I error in algorithm comparison: a lower $p_{\gamma}$ has an associated lower probability of predicting an unfairly longer equivalent runtime.
	To obtain an unbiased prediction of the equivalent runtime, it is enough to consider the parameter $p_{\gamma} = 0.5$.

	To make the computation of the equivalent runtime convenient for the user, we created a stand-alone (no dependencies) python script \href{https://github.com/EtorArza/RTDHW}{equivalent\_runtime.py} available in our GitHub \href{https://github.com/EtorArza/RTDHW}{repository}.
	This script predicts the equivalent runtime with the formula in Definition~\ref{definition:estimation_execution_time}, but considering the parameter $p_{\gamma}$ instead of $\gamma$.
	To achieve this, the $\gamma$ associated to $p_{\gamma}$ is calculated first.
	Given $p_{\gamma}$ the desired probability of estimating a runtime longer than the true equivalent runtime, $s_1,s_2$ the PassMark single-thread score of machines $M_1,M_2$ respectively and $t_1$ the runtime in machine $M_1$, we can use this script to estimate the equivalent runtime.
	For example if $p_{\gamma} = 0.1$, $s_1 = 1540$, $s_2 = 1643$ and $t_1 = 15.0$, then we can get $\hat{t_2}$ with

	\begin{lstlisting}
	python equivalent_runtime.py 0.1 1540 1643 15.0
	>> 11.461295
	\end{lstlisting}

	Even though the proposed model has an arbitrarily low probability of estimating a longer than actual equivalent runtime, this probability is not zero.
	In Section~\ref{section:modifying_sign_test}, we will propose a modification of the sign test that takes into account this probability and avoids an increase in the probability type I error.

	\subsection{Validation}
	\label{sec:verification_passmark_prediction}
	
	We have introduced a methodology to predict the equivalent runtime for single-thread optimization processes---a sequence of computer instructions that can be replicated in different machines---based on the PassMark single-thread score.
	In the following, we experimentally validate that the proposed methodology works as intended.
	To this end, we try to answer the following two questions: i) Is using the equivalent runtime better than using the same runtime in two machines? And ii) does predicting the equivalent runtime for other machines and optimization processes work as intended?

	\vspace{1em}
	\subsubsection{i) Predicting the equivalent runtime vs. using the same runtime}
	In this paper, we proposed predicting the equivalent runtime with the PassMark score.
	In the following, we show that it is better than just using the same runtime in two machines.
	To do so, we compare the prediction error, measured as the ratio with respect to the true equivalent runtime, when using the centered estimator in Equation~\eqref{equation:centered_estimator_runtime} of the equivalent runtime ($\hat{t}_2 = t_1 \cdot \frac{\minusbdivaparamregression - s_2}{\minusbdivaparamregression - s_1}$ )  vs. when estimating it as the same runtime as in the other machine ($\hat{t}_2 = t_1$).

	We estimate the prediction error with these two methods for the 64 optimization processes and 8 machines considered in the calibration of the linear regression (see Appendix~\ref{appendix:experimentation_execution} for details).
	Once we have measured the ratio between the estimated runtime and the true equivalent runtime, we apply the loss function $f(x) = \abs(\text{Log}_2(x))$, obtaining the \textit{log deviation ratio}.
	With this loss function, a prediction that was three times the true equivalent runtime is assigned the same loss as a prediction that was a third of the true equivalent runtime.
	In addition, the log deviation ratio is easier to interpret: for example, a log deviation ratio of $0$ means that the prediction was perfectly accurate, and a log deviation ratio of $1$ denotes that the prediction was double or half the true value etc.

	In Figure~\ref{fig:verif0fitdataequivvsruntimecumulativekernelsize0} we show the empirical distribution function of the log deviation ratio for \textit{equivalent runtime} and \textit{same runtime}.
	The results clearly point out that \textit{equivalent runtime}, predicting the equivalent runtime with the centered estimator in Equation~\eqref{equation:centered_estimator_runtime}, consistently produces a lower (better) error when compared to using the same runtime in two machines \textit{same runtime}.

	\begin{figure}
		\centering
		\includegraphics[width=0.7\linewidth]{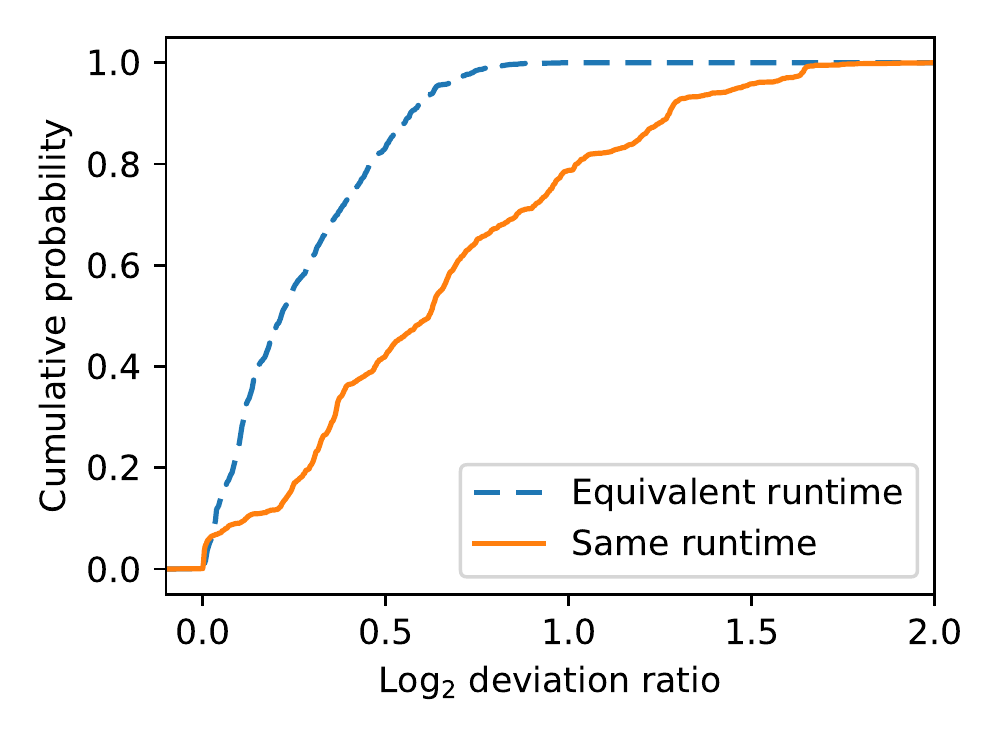}
		\caption{
			A comparison in estimation error of the predicted equivalent runtime (Equivalent runtime) and simply using the same runtime in both machines (Same runtime) with respect to the true equivalent runtime.
			The estimation error is measured as the log deviation ratio of the prediction of the equivalent runtime with respect to the true equivalent runtime. 
			A value closer to 0 indicates a lower prediction error.
		}
		\label{fig:verif0fitdataequivvsruntimecumulativekernelsize0}
	\end{figure}

	\vspace{1em}
	\subsubsection{ii) Validation in other optimization processes and CPUs}
	
	Each time we predict the equivalent runtime with the centered estimator in Equation~\eqref{equation:centered_estimator_runtime}, we expect the prediction to be either higher or lower than the true equivalent runtime.
	Controlling this prediction error is one of the key challenges of the proposed methodology, and allows the user to predict an equivalent runtime with a desired probability of being higher than the true equivalent runtime.
	However, since we fitted this estimator with a set of optimization problems and CPUs (described in detail in Appendix~\ref{appendix:experimentation_execution}), we need to validate the prediction with a different set of CPUs and optimization processes.
	
	\vspace{1em}
	\emph{Validation CPUs and optimization processes}
	
	The four optimization processes for the validation experiment are enumerated below.
	These four optimization processes are very different from the optimization processes used to fit the estimator.
	\begin{enumerate}
		\item Find the first $10^6$ prime numbers.
		\item Finding magic squares~\cite{Charles2013}.
		\item Solving the knapsack problem~\cite{geeksknapsack}.
		\item Solving the N queens problem~\cite{geeksqueens}.
	\end{enumerate}
	
	The CPUs of the machines considered in the validation experiment are listed below:

	{\centering
	\begin{footnotesize}
		\begin{tabular}{lr}
			\textbf{CPU model name}    & \textbf{PassMark score} \\
			Intel(R) Xeon(R) CPU @ 2.20GHz & 1383 \\
			Intel(R) Core(TM)2 Duo CPU P9600 @ 2.53GHz & 1075 \\
			Intel(R) Xeon(R) CPU E5-2680 v3 @ 2.50GHz & 1779 \\
			Intel(R) Xeon(R) Gold 6140 CPU @ 2.30GHz & 1840 \\
			Intel(R) Core(TM) i5-5200U CPU @ 2.20GHz & 1511 \\
		\end{tabular}
	\end{footnotesize}
	}

	Now, we compare the log deviation ratio of the centered estimator (Equation~\eqref{equation:centered_estimator_runtime}) both for the optimization processes and CPUs used to fit the estimator, and these new validation optimization processes and CPUs. 
	The empirical distribution function of the log deviation ratio is shown in Figure~\ref{fig:verif2equivfitdatavsverificationcumulativekernelsize0}.
	The error obtained with the CPUs and optimization processes used to fit the estimator (from now on \textit{Train}) is a lot smoother than with the validation CPUs and optimization processes (from now on \textit{Validation}).
	However, this is to be expected because \textit{Train} contains a larger amount of both optimization processes and CPUs.
	In addition, notice that for most of the $x$-axis, the error of \textit{Validation} is lower (better) than the error of \textit{Train}.
	This can also be explained by the variance of the error \textit{Validation} being larger due to the smaller amount of CPUs and optimization processes.

	In any case, both errors are very similar and close to each other, and this implies that the proposed methodology is indeed applicable to different CPUs and optimization processes.

	\begin{figure}
		\centering
		\includegraphics[width=0.7\linewidth]{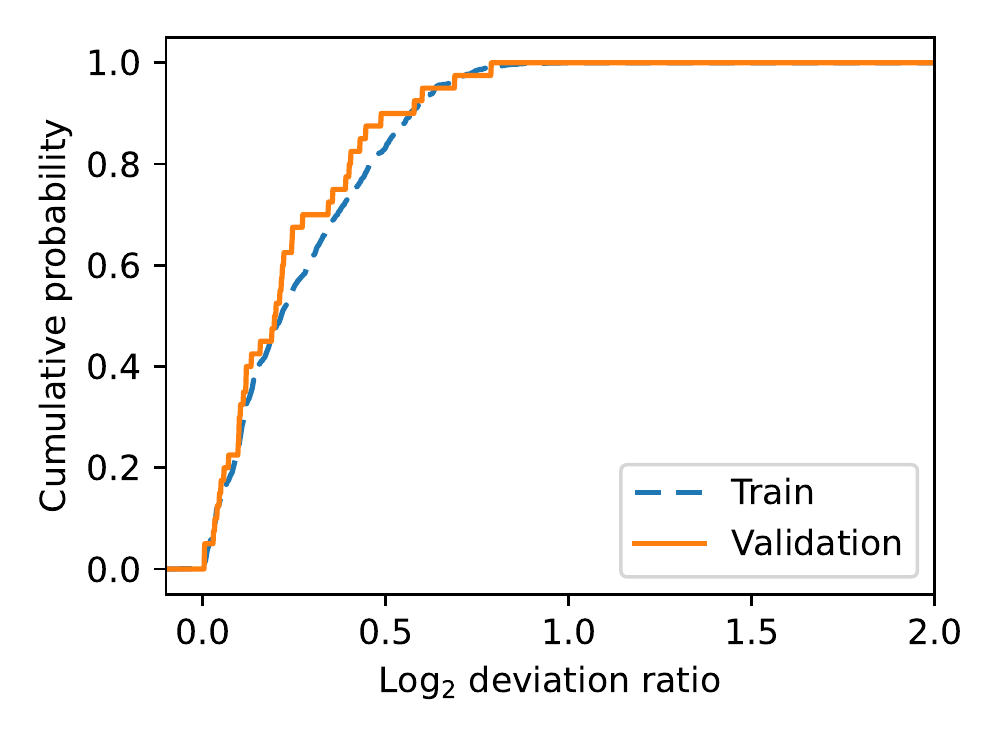}
		\caption{A comparison in estimation error of the equivalent runtime with the centered estimator.
		The estimation error for the optimization processes and CPUs used to fit the estimator (Train), and these new validation optimization processes and CPUs (Validation) are compared.
		The estimation error is measured as the log deviation ratio of the prediction of the equivalent runtime with respect to the true equivalent runtime. 
		A value closer to 0 indicates a lower prediction error.}
		\label{fig:verif2equivfitdatavsverificationcumulativekernelsize0}
	\end{figure}

	\section{Modifying the one-sided sign test}	
	\label{section:modifying_sign_test}


%

	In the previous sections, we proposed an estimator of the equivalent runtime of an algorithm in a machine.
	Specifically, we proposed a biased estimator with an arbitrary probability of estimating a runtime longer than the true equivalent runtime.
	By using this estimator, we can adjust the expected percentage of samples of the performance of algorithm $B$ computed with a runtime longer than the true equivalent runtime.
	When a runtime longer than the true equivalent time is estimated, the probability of making a type I error is higher than if the comparison were carried out in the same machine.
	Therefore, in this section, we propose a correction of the one-sided statistical test that takes into account the probability of estimating a longer than the true equivalent time and its subsequent increase in the probability type I error.

	Given algorithms $A,B$, a one-sided hypothesis test in algorithm performance comparison is as follows\footnote{It is noteworthy that failing to reject $H_0$ does not imply statistical evidence that $H_0$ is true, instead it suggests a lack of evidence against $H_0$. Therefore, in this case, it would not be correct to conclude that ``the algorithms compared perform the same with a statistical significance of $1-\alpha$''. }:

	\[
	\fbox{\begin{minipage}{0.82\textwidth}
			$H_0$: The performance of algorithm  $B$ is worse than or equal to $A$.\\
			$H_1$: The performance of algorithm  $B$ is better than $A$.
	\end{minipage}}
	\]


	We believe that the sign test~\cite{conover1980practical}  is a suitable hypothesis in the context of this paper and, in general, for comparing the performance of optimization algorithms (see Appendix~\ref{appendix:why_sign_test} for details).
	We limit the statistical test to the one-sided sign test, with the alternative hypothesis being that the algorithm whose equivalent runtime was estimated has a higher performance. 
	In the following, we propose a correction for the sign test that does not increase the probability of type I error.

	\subsection{One-sided sign test}
	\label{section:corrected_critical_value}

	The sign test~\cite{conover1980practical} is a special case of the binomial test, for $p=0.5$.
	In the context of algorithm performance comparison, the sign test statistically assesses if the paired performance of two algorithms is the same or not.
	Performing this statistical test involves first executing the optimization algorithms $A$ and $B$ in the same machine, with the same stopping criterion, in a set of $n$ problem instances ($i\in\{1,...,n\}$), obtaining the scores $a_i$ and $b_i$ for each algorithm-instance pair.
	
	These scores depend on which random seed was chosen (this seed represents all the nondeterministic parts of the algorithms, such as solution initialization).
	Thus, the performance of an algorithm in an instance can also be seen as a random variable that is completely determined, given a certain seed.
	We denote the random variables that represent the performance of algorithms $A$ and $B$ in an instance $i$ as $A_i$ and $B_i$, respectively.
	
	The statistical test allows us to draw conclusions about the algorithms on a larger set of problem instances based on the observed results in the set of $n$ instances.	
	The sign test is based on these three assumptions~\cite{conover1980practical}:
	
	\begin{itemize}
		\item Each of the sample pairs $A_i,B_i$ are mutually independent of the rest of the pairs.
		\item Any observable pair $a_i,b_i$ can be compared, that is, we can say that $a_i < b_i$, $b_i < a_i$ or $a_i = b_i$.
		\item The pairs are internally consistent, or if $\mathcal{P}[A_i > B_i] > \mathcal{P}[A_i < B_i]$ for one pair, then the same is true for all pairs. 
	\end{itemize}	
	
	In the context of algorithm comparison, the most problematic assumption is the first one.
	The reason is that in real-life benchmarks, some problem instances may share similarities, which means that there is no complete independence among all sample pairs $A_i,B_i$.
	The Mann-Whitney and the Wilcoxon signed rank test also contain this first assumption~\cite{conover1980practical}.
	However, in practice, this limitation is usually ignored.
	This is why, in general, it is a good idea to use a set of benchmark problems with many kinds of different instances.
	As future work, the proposed methodology could be adapted to be applicable to multiple benchmark sets, where the instances in each benchmark have similar properties.

	From now on, without loss of generality\footnote{A maximization problem can be converted into a minimization problem by multiplying the objective function with $-1$.}, we assume that the algorithms deal with a minimization problem, (i.e., $a_i$ is better than $b_i$ $\iff$ $a_i<b_i$). 
	We define $\#\{A_i < B_i\}$ as a random variable that counts the number of cases\footnote{Without loss of generality, we can assume that $a_i \neq b_i$, because samples in which $a_i = b_i$ are removed when performing the sign test~\cite{conover1980practical}.} that $A_i<B_i$ in $n$ instances.
	Then, the following hypothesis test corresponds to the one-sided sign test~\cite{conover1980practical}:

	\[
	\fbox{\begin{minipage}{0.4\textwidth}\centering
			$H_0: \mathcal{P}[A_i < B_i] \geq \mathcal{P}[A_i > B_i]$\\
			$H_1: \mathcal{P}[A_i < B_i] < \mathcal{P}[A_i > B_i]$
	\end{minipage}}
	\]

	Under $H_0$, the null distribution of $\# \{A_i < B_i\}$  is $Bin(n,0.5)$, where $Bin(n,0.5)$ is the binomial distribution of size $n$ and rate of success $0.5$~\cite{conover1980practical}.
	The \textit{p}-value for this hypothesis test is 
	
	\begin{equation}
	p(k) = \mathcal{P}[\# \{A_i < B_i\} \leq k \lwhere H_0] = \mathcal{P}[Bin(n,0.5) \leq k]
	\end{equation}
	
	where $k$ is substituted by the statistic of the sign test: the number of cases that $a_i<b_i$ in all $i\in \{1,...,n\}$ samples, denoted as $\#\{a_i < b_i\}$.
	By definition~\cite{wasserstein_asa_2016}, the \textit{p}-value can be interpreted as the probability of obtaining a more extreme (lower) statistic than the observed, assuming $H_0$ is true. 
	If we reject the null hypothesis when $p(\#\{a_i < b_i\}) \leq \alpha$, then the probability of type I error is less than or equal to $\alpha$.

	\subsection{The corrected p-value}

	In practice, the statistic $\#\{a_i < b_i\}$ cannot be computed because the true equivalent runtime $t_2$ is unobservable.
	The equivalent runtime is approximated with $\hat{t}_2$ (see Definition~\ref{definition:estimation_execution_time}).
	As a result, each $b_i$ is substituted with its corresponding $\hat{b}_i$, which is computed by using $\hat{t}_2$ instead of $t_2$ as the stopping criterion.
	This means that the statistic $\#\{a_i < b_i\}$ is replaced by $\#\{a_i < \hat{b}_i\}$, which counts the number of times that $a_i < \hat{b}_i$ (without loss of generality, minimization is assumed) is observed.
	Therefore, we need to define the function to obtain the \textit{p}-value associated to the statistic $\#\{a_i < \hat{b}_i\}$:
	
	\begin{equation}
	\hat{p}(k) = \mathcal{P}[\# \{A_i < \hat{B}_i\} \leq k \lwhere H_0]
	\end{equation}
	
	where the \textit{p}-value is obtained by substituting $k$ with the observed statistic $\#\{a_i < \hat{b}_i\}$.
	The \textit{p}-value is the probability of obtaining a statistic that is lower than or equal to the observed when $H_0$ is true.

	Notice that if we reject the null hypothesis when $\hat{p}(\#\{a_i < \hat{b}_i\}) \leq \alpha$, then the probability of type I error is less than or equal to $\alpha$.
	However, for this to hold, we need to assume that $\hat{b}_i < b_i$: a longer optimization time produces a lower or equal objective value (in a minimization setting).
	In general, we assume that a longer optimization time can only produce a lower or equal objective value.

	Let $p_\gamma$ be the desired upper bound of the probability of predicting a runtime longer than the true equivalent runtime for each instance $i$.
	Then, in more than $1 - p_\gamma$ of cases, $b_i$ is obtained with a longer runtime than $\hat{b}_i$, and, therefore, the probability of $\hat{b}_i \geq b_i$ is greater than $1 - p_\gamma$.
	When a small $p_\gamma$ is chosen, we expect that $\#\{a_i < \hat{b}_i\}$ is higher than or equal to $\#\{a_i < b_i\}$, but it will not always be so.
	To overcome this limitation, we need to define a corrected \textit{p}-value $\hat{p}_c$, an upper bound of the actual \textit{p}-value associated with statistic $\#\{a_i < \hat{b}_i\}$, that takes into account the probability $p_\gamma = \hat{t}_2 > t_2$.
	Specifically, we define this upper bound as 
	
	\begin{equation}
	\label{equation:compute_corrected_p_value}
	\hat{p}_c(k) = \sum_{v=0}^{n} (1 - \mathcal{P}[Bin(n,p_\gamma) < \max(0,v - k)]) \mathcal{P}[Bin(n,0.5) = v]
	\end{equation}
	
	 such that 
	
	\begin{equation}
	\label{equation:p_value_bound_sign_test}
	\hat{p}_c(k) > \mathcal{P}[\#\{A_i < \hat{B}_i\} \leq k \lwhere H_0] = \hat{p}(k)
	\end{equation}

	is satisfied (we prove this inequality in Appendix~\ref{appendix:proof_inequality_p_corrected}), where $H_0$ implies that statistic $\#\{A_i < B_i\}$ follows the null distribution $Bin(n,0.5)$~\cite{conover1980practical}.
	Figure~\ref{fig:alpha_vs_corrected_alpha} shows $p$ and $\hat{p}_c$ side by side. 
	Notice that $\hat{p}_c$ is slightly higher, because it needs to account for the probability that $\hat{t}_2 > t_2$.
	The corrected \textit{p}-value $\hat{p}_c$ is interesting because rejecting $H_0$ when $\hat{p}_c(\#\{a_i < \hat{b}_i\}) < \alpha$ has also an associated probability of type I error lower than $\alpha$.
	The reason is that $\hat{p}_c(\#\{a_i < \hat{b}_i\}) > \hat{p}(\#\{a_i < \hat{b}_i\})$, and therefore, $ \hat{p}_c(\#\{a_i < \hat{b}_i\}) < \alpha \Rightarrow \hat{p}(\#\{a_i < \hat{b}_i\}) < \alpha$.

	\begin{figure}
		\centering
		\caption*{$\hspace{3.2em}$ \textit{p}-value for the sign test}
		\includegraphics[width=0.7\linewidth]{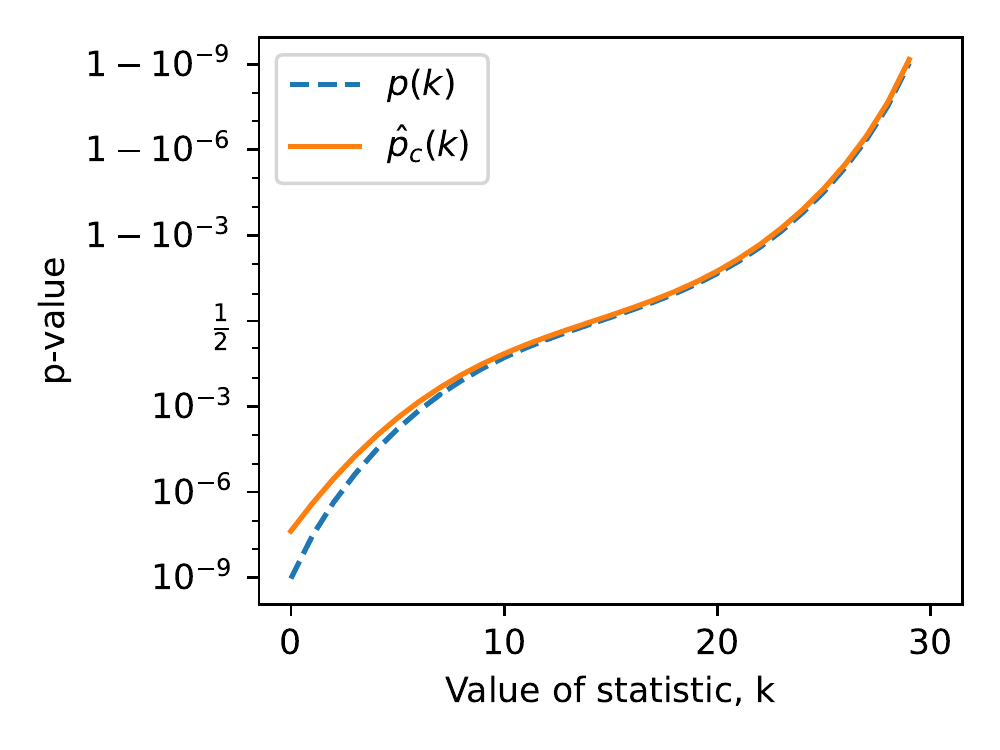}
		\caption
		{
			This figure shows the \textit{p}-value and the corrected \textit{p}-value $\hat{p}_c$ for the sign test with a sample size of $n = 30$ and $p_\gamma = 0.01$.
			Under the null hypothesis $H_0$, the \textit{p}-value represents the probability of $\#\{a_i < b_i\} \leq k$, while the corrected \textit{p}-value represents an upper bound of the probability of $\#\{a_i < \hat{b}_i\} \leq k$.
			Under the null hypothesis, $\#\{a_i < b_i\}$ follows a binomial distribution of size $n$ and probability of success $0.5$.
		}
		\label{fig:alpha_vs_corrected_alpha}
	\end{figure}

	To generate the corrected p-values conveniently, we created a standalone (only dependencies in the standard library) python script \href{https://github.com/EtorArza/RTDHW/blob/master/corrected_p_value.py}{corrected\_p\_value.py} available in our \href{https://github.com/EtorArza/RTDHW}{GitHub repository}.
	This script uses an efficient and precise implementation of Equation~\eqref{equation:compute_corrected_p_value}.
	To calculate the corrected p-value, we need the probability of predicting a runtime longer than the true equivalent runtime $p_\gamma$, the sample size $n$, and the number of observations $k$ in which algorithm $A$ outperforms algorithm $B$.
	For example if $p_{\gamma} = 0.1$, $n = 20$, and $k = 3$, then we can get $\hat{p}_c(3)$ with
	
	\begin{lstlisting}
	python corrected_p_value.py 0.1 20 3
	>> 0.043596000
	\end{lstlisting}

	\section{Applying the methodology}
	\label{section:case_study}

	In this section, we illustrate how to apply the proposed methodology with two examples.
	The proposed methodology is very simple to use and does not require any additional software.
	Further details and material are available in our \href{https://github.com/EtorArza/RTDHW}{GitHub repository}.

	\subsection{Example I}

	In this example, we will compare a simple random initialization local search procedure with a memetic search algorithm by Benlic et al.~\cite{benlic_memetic_2015} for the QAP.
	Benlic et al. run the code sequentially, without any parallel or multithreaded execution.
	Using the proposed methodology, we will statistically assess the difference in the performance between these two algorithms, without having to execute the code of the memetic algorithm.
	In this case, the memetic search algorithm is algorithm $A$, because this is the algorithm of which we already have the results, and the local search algorithm is algorithm $B$, because this is the algorithm whose runtime is going to be estimated.
	For this experiment, we choose a probability of predicting a longer than true equivalent runtime of $p_\gamma= 0.01$.

	\underline{Step 1: Obtaining the data}

	To apply the proposed methodology, we need to find certain information about the execution of the memetic algorithm.
	The required data includes the list of instances to be used in the experimental comparison, the average objective value obtained by the memetic search algorithm, and the runtime of the memetic search algorithm in each of the instances.
	The information extracted from the article by Benlic et al.~\cite{benlic_memetic_2015} is listed in the first three columns of Table~\ref{table:comparison_with_memetic_data_table}.
	Also, we need to find the CPU model of the machine in which the memetic search was run (machine $M_1$), which is "Intel Xeon E5440 2.83GHz" as specified in their article.
	Finally, the machine score of this CPU, measured as PassMark single-thread score, is $s_1 = \singlethreadscoreofXeonEFiveFourFourZero$ (as seen on the file \href{https://github.com/EtorArza/RTDHW/blob/master/cpu_scores.md}{cpu\_scores.md} in the GitHub repo).

		\begin{table}[]
		\centering
		\begin{footnotesize}
			\caption*{\small }
			\begin{tabular}{l|r|r|r|r}
				\hline
\multicolumn{1}{c}{ } & \multicolumn{2}{c|}{\multirow{2}{3.7cm}{\textbf{Data obtained from the paper by Benlic et al.~\cite{benlic_memetic_2015}}}} & \multicolumn{2}{c}{\multirow{2}{4.5cm}{\textbf{Data corresponding to the execution of $B$ in machine $M_2$}}} \\
\multicolumn{1}{c}{ } & \multicolumn{2}{c|}{\textbf{}}                               \\ \hline
\multirow{2}{4em}{instance} &     runtime in		  &    objective & estimated equivalent & objective \\
							& seconds, $t_1$ & value, $a_i$ & runtime, $\hat{t}_2$ &       value, $\hat{b}_i$ \\ \hline
				tai40a      &                     486 &       \textbf{       3141222} &  313.68     &     3207604 \\
				tai50a      &                    2520 &       \textbf{       4945266} &  1626.50    &     5042830 \\
				tai60a      &                    4050 &       \textbf{       7216339} &  2614.02    &     7393900 \\
				tai80a      &                    3948 &      \textbf{       13556691} &  2548.18    &    13840668 \\
				tai100a     &                    2646 &      \textbf{       21137728} &  1707.82    &    21611122 \\
				tai50b      &                      72 &     \textbf{       458821517} &  46.47      &   459986202 \\
				tai60b      &                     312 &     \textbf{       608215054} &  201.37     &   609946393 \\
				tai80b      &                    1878 &     \textbf{       818415043} &  1212.13    &   824799510 \\
				tai100b     &                     816 &    \textbf{       1185996137} &  526.67     &  1195646366 \\
				tai150b     &                    4686 &     \textbf{       499195981} &  3024.52    &   505187740 \\
				sko100a     &                    1338 &        \textbf{       115534} &  863.59     &      153082 \\
				sko100b     &                     390 &        \textbf{       152002} &  251.72     &      155218 \\
				sko100c     &                     720 &        \textbf{       147862} &  464.71     &      149076 \\
				sko100d     &                    1254 &        \textbf{       149584} &  809.37     &      150568 \\
				sko100e     &                     714 &        \textbf{       149150} &  460.84     &      150638 \\
				sko100f     &                    1380 &        \textbf{       149036} &  890.70     &      150006
			\end{tabular}
		\end{footnotesize}
		\caption
		{
			This table shows all the data in the first example.
			The first three columns correspond to the QAP instances in which the memetic search algorithm by Benlic et al.~\cite{benlic_memetic_2015} was tested, the runtime of the memetic search algorithm in each instance, and the best objective value they obtained in each execution, averaged in 10 executions per instance.
			The information in these three columns was directly obtained from the paper by Benlic et al., without any additional executions.
			The next two columns show the estimated equivalent runtimes and the average objective value scores that the local search algorithm obtained with this runtime as the stopping criterion.
			The local search algorithm was executed in machine $M_2$.
		}
		\label{table:comparison_with_memetic_data_table}
	\end{table}

	\underline{Step 2: Estimating the equivalent runtime}
	
	With the data already gathered, the next step is to estimate the equivalent runtime of each instance for the machine in which the local search algorithm will be executed (machine $M_2$).
	Estimating the runtime requires the score $s_2$ of this machine to be known.
	The CPU model of $M_2$ is "Intel Celeron N4100", with a PassMark single-thread score of $s_2 = \singlethreadscoreofCeleronNFourOneZeroZero$.
	With this information, we are ready to estimate the equivalent runtime $\hat{t}_2$ for each instance in machine $M_2$.
	We run the script
	
	\begin{lstlisting}
	python equivalent_runtime.py 0.01 1219 1012 t_1
	\end{lstlisting}

	where $t_1$ is substituted with the runtime of the memetic search algorithm in each instance, listed in Table~\ref{table:comparison_with_memetic_data_table}.

	\underline{Step 3: Running the experiments}

	Now, we execute the local search algorithm in the instances listed in Table~\ref{table:comparison_with_memetic_data_table}, using the estimated runtimes $\hat{t}_2$ as the stopping criterion.
	This execution is carried out in machine $M_2$, and the best objective function values $\hat{b}_i$ are listed in Table~\ref{table:comparison_with_memetic_data_table}.
	Following the procedure by Benlic et al., these best objective values are averaged over 10 executions.

	\underline{Step 4: Obtaining the corrected p-value}

	Once all the results have been computed, the next step is to compute the statistic $\#\{a_i < \hat{b}_i\}$, which counts the number of times that $a_i < \hat{b_i}$.
	In this case, $a_i < \hat{b_i}$ happens 15 times, and therefore, $k = \#\{a_i < \hat{b}_i\} = 15$.
	Now we can compute the corrected p-value of the one sided sign test.
	To do so, we use the script \href{https://github.com/EtorArza/RTDHW/blob/master/corrected_p_value.py}{corrected\_p\_value.py} with the chosen $p_\gamma= 0.01$, $n = 15$ and $k = 15$.
	
	\begin{lstlisting}
	python corrected_p_value.py 0.1 15 15
	>> 1.0000000
	\end{lstlisting}

	\underline{Step 5: Conclusion}
	
	Since the observed corrected p-value is \underline{not} lower than the chosen $\alpha = 0.05$, we cannot reject $H_0$.
	In this case, the conclusion is that with the amount of data that we have and the chosen target probability of type I error of $\alpha = 0.05$, we can not say that the local search algorithm has a statistically significantly better performance than the memetic search algorithm\footnote{It would not be correct to conclude that the two algorithms perform (statistically significantly) the same, or that the memetic search performs statistically significantly better than the local search.}.

	It is important to note that, if we had considered the original runtimes $t_1$ as the stopping criterion for algorithm $B$ in machine $M_2$ (longer than the estimated equivalent runtime $\hat{t}_2$), the local search would have had an unfairly longer runtime.
	In other words, the comparison would have been biased toward the local search.

	\subsection{Example II}
	
	In this second example, we will compare the same simple random initialization local search procedure with an estimation of distribution algorithm (EDA) for the QAP~\cite{arza_kernels_2020}.
	The estimation of distribution algorithm (EDA) was executed sequentially, without parallel or multithreaded execution.
	In this case, the EDA is algorithm $A$, because this is the algorithm of which we already have the results, and the local search algorithm is algorithm $B$, because this is the algorithm whose runtime is estimated.
	For this experiment, we choose a probability of predicting a longer than true equivalent runtime of $p_\gamma= 0.01$.
	
	\pagebreak
	\underline{Step 1: Obtaining the data}
	
	To apply the proposed methodology, we need to find certain information about the execution of the EDA.
	The required data includes the list of instances to be used in the comparison, the average objective value obtained by the EDA, and the runtime used in each instance.
	The information extracted from the paper~\cite{arza_kernels_2020} is listed in Table~\ref{table:comparison_with_EDA_table}.
	In addition, we need to find the CPU model of the machine in which the EDA search was run (machine $M_1$), which is "AMD Ryzen 7 1800X", as specified in the paper.
	Finally, the machine score of this CPU, measured as PassMark single-thread score is $s_1 = 2182$, as seen on the \href{https://github.com/EtorArza/RTDHW/blob/master/cpu_scores.md}{cpu\_scores.md} file.

	\begin{table}[]
		\centering
		\begin{footnotesize}
			\begin{tabular}{l|r|r|r|r}
	\hline
\multicolumn{1}{c}{ }	& \multicolumn{2}{c|}{\multirow{2}{3.7cm}{\textbf{Data obtained from the paper by Arza et al.~\cite{arza_kernels_2020}}}} & \multicolumn{2}{c}{\multirow{2}{4.5cm}{\textbf{Data corresponding to the execution of $B$ in machine $M_2$}}} \\
\multicolumn{1}{c}{ }	& \multicolumn{2}{c|}{\textbf{}}     \\ \hline
	\multirow{2}{4em}{instance} &     runtime in		  &    objective & estimated equivalent & objective \\
	& seconds, $t_1$ & value, $a_i$ & runtime, $\hat{t}_2$ &       value, $\hat{b}_i$ \\ \hline
				bur26a     &      1.45 &       5432374 	 		& 1.80  & \textbf{5426670} \\
				bur26b     &      1.45 &       3824798 	 		& 1.80  & \textbf{3817852} \\
				bur26c     &      1.43 &       5427185 			& 1.77  & \textbf{5426795} \\
				bur26d     &      1.44 &       3821474 			& 1.78  &  \textbf{3821239}\\
				nug17      &      0.44 &          1735 			& 0.54  &  \textbf{1734}\\
				nug18      &      0.51 &          1936 			& 0.63  &  1936\\
				nug20      &      0.68 &          2573 			& 0.84  & \textbf{2570} \\
				nug21      &      0.77 &          2444 			& 0.95  & 2444\\
				tai10a     &      0.12 &        135028 			& 0.14  &135028   \\
				tai10b     &      0.12 &       1183760 			& 0.14  & 1183760 \\
				tai12a     &      0.18 &        224730 			& 0.22  & \textbf{224416} \\
				tai12b     &      0.19 &      39464925 			& 0.23  & 39464925 \\
				tai15a     &      0.31 &        388910 			& 0.38  & \textbf{388214} \\
				tai15b     &      0.31 &      51768918 			& 0.38  & \textbf{51765268} \\
				tai20a     &      0.69 &      709409            & 0.85  & \textbf{703482} \\
				tai20b     &      0.68 &     122538448 			& 0.84  & \textbf{122455319} \\
				tai40a     &      5.41 &     \textbf{3194672}   & 6.72  &  3227894\\
				tai40b     &      5.41 &     644054927 			& 6.72  &   \textbf{637470334}\\
				tai60a     &     19.23 &     \textbf{7367162}   & 23.88 & 7461354 \\
				tai60b     &     19.21 &   \textbf{611215466}   & 23.86 & 611833935 \\
				tai80a     &     50.09 &    \textbf{13792379}   & 62.22 & 13942804 \\
				tai80b     &      50.1 &     836702973 		    & 62.23 & \textbf{830729983}

			\end{tabular}
		\end{footnotesize}
		\caption
		{
			This table shows all the data in the second example.
			The first three columns correspond to the QAP instances in which the EDA algorithm by Arza et al.~\cite{arza_kernels_2020} was tested, the runtime of the EDA in each instance, and the best objective value they obtained in each execution, averaged in 10 executions per instance.
			The information in these three columns was directly obtained from this paper, without any additional executions.
			The next two columns show the estimated equivalent runtimes and the average objective value scores that the local search algorithm obtained with this runtime as the stopping criterion.
			The local search algorithm was executed in machine $M_2$.
		}
		\label{table:comparison_with_EDA_table}
		
	\end{table}

	\pagebreak
	\underline{Step 2: Estimating the equivalent runtime}
	
	With the data already gathered, the next step is to estimate the equivalent runtime of each instance for the machine in which the local search algorithm will be executed (machine $M_2$).
	The CPU model of $M_2$ is "Intel Celeron N4100" (the same as in the previous example), with a PassMark single-thread score of $s_2 = \singlethreadscoreofCeleronNFourOneZeroZero$.
	To estimate the runtime for each instance, we run
	
	\begin{lstlisting}
	python equivalent_runtime.py 0.01 2182 1012 t_1 
	\end{lstlisting}

	where $t_1$ is substituted with the runtime of the EDA algorithm in each instance, listed in Table~\ref{table:comparison_with_EDA_table}.

	\underline{Step 3: Running the experiments}
	
	Now, we execute the local search algorithm in the instances listed in Table~\ref{table:comparison_with_EDA_table}, using the estimated runtimes $\hat{t}_2$ as the stopping criterion.
	This execution is carried out on machine $M_2$, and the best objective function values $\hat{b}_i$ are listed in Table~\ref{table:comparison_with_EDA_table}.
	Following the procedure by Arza et al., these best objective values are averaged over 20 executions.

	\underline{Step 4: Obtaining the corrected p-value}

	After the executions, the statistic $k = \#\{a_i < \hat{b}_i\}$ is computed, which counts the number of times that $a_i < \hat{b_i}$.
	In this case, $a_i < \hat{b_i}$ happens 4 times, and therefore, $\#\{a_i < \hat{b}_i\} = 4$.
	Now we compute the corrected p-value with with the chosen $p_\gamma= 0.01$, $n = 17$ and $k = 4$.
	
	\begin{lstlisting}
		python corrected_p_value.py 0.01 17 4
		>> 0.033192784
	\end{lstlisting}

	\underline{Step 5: Conclusion}
	
	The observed corrected p-value is lower than the chosen $\alpha = 0.05$, and therefore we reject $H_0$.
	The conclusion is that with a probability of type I error of $\alpha = 0.05$, the performance of the local search procedure is statistically significantly better than the performance of the EDA.
	
	In this case, machine $M_1$ is more powerful (in terms of computational capabilities) than machine $M_2$. 
	If we had considered the original runtimes $t_1$ as the stopping criterion for algorithm $B$ in machine $M_2$ (shorter than the estimated equivalent runtime $\hat{t}_2$), it would have been more difficult for the local search to perform better than the EDA.
	In that case, $H_0$ might not have been rejected.

\section{Limitations, applicability and future work}
\label{section:limitations}

Below, we discuss the limitations, applicability, and potential future developments of the proposed methodology.

\subsection{Multiple threads/cores}

The purpose of the presented work is to compare two algorithms with the same computational resources.
One of the limitations of the presented model is that it should only be applied with optimization processes that run on a single-thread.
This is because the single-thread PassMark score and the optimization process $\rho'$ (see Appendix~\ref{appendix:experimentation_execution})) that were used to calibrate the linear regression in Definition~\ref{definition:linear_regression_passmark} are also single-threaded.

However, many of the optimization algorithms in the literature today are not single-threaded.
For example, many problems involve linear algebra operations that can benefit from a multithreaded speedup, and most consumer CPUs today have at the very least two cores.
Hence, the single-thread PassMark score is not suitable for optimization algorithms that involve these types of operations.
Although, theoretically, it should always be possible to execute parallel code sequentially.

The multi thread \href{https://www.cpubenchmark.net/high_end_cpus.html}{PassMark score} does take into account the multithreaded nature of the CPUs, and could, therefore, be used to re-calibrate the linear regression.
This re-calibration would also involve defining another optimization process $\rho'$ that makes use of the multi thread capabilities of the CPU, such as solving linear problems or other processes that involve matrix multiplications.

There is, however, an additional limitation inherent to parallel executing algorithms that makes their comparison in different machines difficult. 	
Suppose we have two algorithms that run in parallel and their maximum speedup is obtained when executed in four cores, and additional threads/cores offer negligible improvement.
Now let us assume that we have two machines, $M_1$ with four cores and $M_2$ with sixteen.
Let us also assume that the cores in these two machines are similar in speed.
Then, roughly speaking, we expect that the algorithm executed in machine $M_2$ should have an equivalent runtime 4 times shorter.
However, since the two algorithms can only take advantage of the speedup provided by, at most, four cores; the algorithm executed in machine $M_2$ would have a huge disadvantage.

In general, it is difficult to know the maximum potential speedup of the execution in parallel of an algorithm presented in the literature.
This makes the prediction of the equivalent runtime of parallel algorithms more challenging than for single-thread algorithms.
Taking into account the additional difficulty associated to the comparison of parallel algorithms, we think that the comparison of optimization algorithms that run on single-thread is a reasonable starting point.
In future work, it would be interesting to extend the proposed methodology for multi thread algorithms.

\subsection{CPU as the only bottleneck} 

The PassMark single-thread score measures the computing capability of the CPU, disregarding other components like the hard disk or the speed and the size of the RAM memory.
Therefore, the prediction of the equivalent runtime is only applicable to optimization algorithms that are CPU intensive, or in other words, optimization algorithms that have their execution speed limited by the CPU.

There are many optimization algorithms whose runtime is determined mainly by the speed of the memory, instead of the CPU.
Specially, when optimization algorithms require large amounts of data to be loaded to memory repeatedly.
Conversely, when an optimization algorithm does not use too much memory, we can expect its runtime to be more CPU dependent.

In addition, many optimization algorithms in machine learning today are executed in GPUs and sometimes even on specialized hardware.
Compared to CPU execution, GPU offers speedups when similar operations are applied to large amounts of data.
As in the multi core case, predicting the runtime of algorithms in GPUs is more challenging than in single-threaded execution in CPU.
GPUs themselves have several processing cores and integrated memory that varies in size and speed from model to model.

The proposed methodology could be adapted for algorithms whose runtime depends on memory or GPU speed.
In fact, PassMark has a benchmarks for RAM and GPUs.
Therefore, it could be possible to re-calibrate the linear regression for either RAM or GPU dependent tasks.
This re-calibration would also involve defining another optimization process $\rho'$.
Taking into account the additional difficulty associated to measuring the runtime of algorithms that depend on RAM and GPU, we think that the comparison of optimization algorithms whose runtime depends primarily on the CPU is a reasonable starting point for this paper.
In future work, it could be interesting to adapt the methodology for algorithms whose runtime depends on RAM or GPU.

\subsection{Efficiency of the implementation}

In addition to the limitations related to the hardware, the implementation of the algorithms can also have an impact in the runtime.
For instance, if the same algorithm is implemented in both Python and C, the runtime in C is likely to be shorter.
But even within the same programming language, the runtime could change depending on the compiler flags (i.e. the -O3 will probably outperform no optimizations) or the configuration of the interpreter.
In addition to the previous factors, the implementation itself could also be more or less efficient, depending on the skill of the programmer and the time it invests in designing efficient code.

Even though there are quite a few factors that depend on the implementation, we argue that by implementing the code in the same programming language the results should be comparable.
In any case, even when the methodology proposed in this paper is not used, it is the responsibility of the researcher to make sure that the comparison is fair in terms of the implementation.
This limitation is not inherent to the proposed methodology, but to the comparison of two algorithms in general.

\subsection{The variance in the PassMark single-thread score}

The PassMark single-thread score is a score for CPUs that is correlated with their single-thread performance.
However, the performance of CPUs is not the same even within the same model. 
This is known as the \textit{silicon lottery}, and is caused by the manufacturing process of CPUs.
In addition, the performance of the CPU will also be limited by the cooling system used.
With better cooling, we can expect a better CPU performance.

The PassMark single-thread score takes into account this variance, and the scores are the averages of several users' submitted results.
Still, the cooling setup and the silicon lottery of the researcher that wants to apply the proposed methodology will inevitably introduce a variance to the predictions of the equivalent runtime.

The presented method models the probability of predicting an equivalent runtime that is longer than the true equivalent runtime. And by doing so, it takes into account this variance because the machines used in the calibration process of the linear regression inevitably have different cooling systems and are also affected by the silicon lottery.

\subsection{Very high and low PassMark scores}

Finally, there is a limitation regarding the chosen machine score: the PassMark single-thread score.
In Section~\ref{section:norm_exec_time}, we saw that a linear function is a suitable function to model the relationship between the machine score and the runtime of the reference optimization process $\rho'$ (the definition of $\rho'$ is given in Appendix~\ref{appendix:experimentation_execution}).
The reference optimization process $\rho'$ is used to calibrate the linear regression in Definition~\ref{definition:linear_regression_passmark} so that the equivalent runtime of other optimization processes $\rho$ can later be predicted based on this formula.
The formula of the fitted linear regression is $t(M_j,\rho') \approx \aparamregression s_j + \bparamregression$ where $t(M_j,\rho')$ is the equivalent runtime of $\rho'$ in machine $M_j$, and $s_j$ is the score of machine $M_j$. 
With this formula, a PassMark single-thread score higher than 3223 produces a negative estimated equivalent runtime, which does not make sense.
However, for the \nmachines\ machines used to fit the data, as Figure~\ref{fig:passmark_base_algorithm_regression} shows, the linear model seems to be suitable.

To overcome this limitation, we recommend applying the proposed methodology only in machines with PassMark single-thread scores in the interval $(411,2185)$.
These values correspond to the highest and lowest values used in the fitting of the linear regression.
More than 70\% of the CPUs in the provided list (see the file \href{https://github.com/EtorArza/RTDHW/blob/master/cpu_scores.md}{cpu\_scores.md} in the GitHub repo) have their PassMark score in this interval.
In addition, the CPUs that do not have their score in this interval are either very new or very old, which means that the proportion of the user base with the PassMark single-thread score in this interval is probably way higher than 70\%.
As future work, and especially when more powerful processors are available, the methodology can be updated to incorporate these new processors or even change the machine score to other benchmark scores beyond the PassMark single-thread score.

\subsection{Assumptions of the corrected sign test}
\label{section:limitations_sign_test}

The corrected sign test is based on certain assumptions and should only be used taking into account certain limitations that will be addressed in this section.
First, we will address the assumption related to the probability of predicting a runtime longer than the true equivalent runtime.
Let $(\hat{t}_2 > t_2)_i$ be a random variable that represents if the estimated equivalent runtime for algorithm $A$, instance $i$ in machine $M_2$, is longer than the true equivalent runtime or not.
The required assumption is similar to assuming that $\hat{t}_2 > t_2$ is mutually independent for each instance $i$.
Specifically, it is required\footnote{This assumption is required in the proof of Equation~\eqref{equation:p_value_bound_sign_test} in Appendix~\ref{appendix:proof_inequality_p_corrected}. Specifically, it is used in Lemma~\ref{lemma:bound_coef_of_theorem}.} that $\mathcal{P}[(\hat{t}_2 > t_2)_i | \cap_{i \neq j} (\hat{t}_2 > t_2)_j] < p_{\gamma}$.

One could argue that this assumption is false because $(\hat{t}_2 > t_2)$ depends on many factors, such as the machines used in the experimentation.
The same two machines are used to compute all samples $a_i,b_i$ suggesting that all $(\hat{t}_2 > t_2)_i$ can never be truly independent among each other.
However, even though we can not ensure that $\mathcal{P}[(\hat{t}_2 > t_2)_i | \cap_{i \neq j} (\hat{t}_2 > t_2)_j] < p_{\gamma}$, by choosing a suitable correction coefficient $\gamma$, in Section~\ref{section:norm_exec_time}, we estimated that $\mathcal{P}[(\hat{t}_2 > t_2)_i] < p_{\gamma}$.

In addition to the previous assumption, the proposed methodology only considers one side hypothesis testing.
In this regard, it should only be applied to show a statistically significantly superior performance of the algorithm whose equivalent runtime was estimated (denoted as algorithm $B$ in this paper).
The reason is that algorithm $B$ has a high probability of having a lower runtime, thus, it is easy that $B$ performs worse than $A$, while the opposite is difficult.
Failing to reject $H_0$ only indicates a lack of evidence against $H_0$, and in our context, it only indicates that there is not enough evidence to say that $B$ performs better than $A$ (it tells us nothing about $A$ performing better than $B$).

	\section{Conclusion}	
	\label{section:conclusion}

	Usually, comparing optimization algorithms with a maximum runtime as a common stopping criterion requires the algorithms to be executed in the same machine.
	Unfortunately, the code of all the algorithms is not always available.
	An alternative is to adjust the runtime of the algorithms relative to the speed of the machine in which they are executed.
	In this paper, we proposed a methodology to statistically compare the performance of two optimization algorithms in two different machines, when the results of one of the algorithms are already known and without having to execute this algorithm again.
	The methodology ensures that the probability of type I error does not increase due to the algorithms being executed in different machines.
	To achieve this, first, the runtime of the executed algorithm is adjusted based on the speed of the CPUs of both machines.
	Then, a modified one-sided sign test is used so that the probability of using an unfairly longer runtime is taken into account.
	We illustrate the application of the proposed methodology with two examples.
	
	Alongside this paper, a tutorial with examples is presented in our \href{https://github.com/EtorArza/RTDHW}{GitHub repository}.
	In addition, we offer two standalone scripts (also in the same repository): one to estimate the equivalent runtime and another one to generate the corrected p-values.
	This will hopefully make it simple for people to apply the proposed methodology.





\section*{Supplementary Material}

\emph{Code to Reproduce the Results}

The code to reproduce the results in this paper are available in the \href{https://github.com/EtorArza/RTDHW}{GitHub repository}.
	
\vspace{1em}

\emph{Scripts to Apply the Methodology}

The scripts \href{https://github.com/EtorArza/RTDHW/blob/master/equivalent_runtime.py}{equivalent\_runtime.py} and \href{https://github.com/EtorArza/RTDHW/blob/master/corrected_p_value.py}{corrected\_p\_value.py} required to apply the methodology are also available in the \href{https://github.com/EtorArza/RTDHW}{GitHub repository}.

\section*{Acknowledgments}	
	This work was funded in part by the Spanish Ministry of Science, Innovation and Universities through PID2019-1064536A-I00 and the BCAM Severo Ochoa excellence accreditation SEV-2017-0718; by Basque Government through consolidated groups 2019-2021 IT1244-19, ELKARTEK program and BERC 2018-2021 program; and by the Spanish Ministry of Economy and Competitiveness through the project TIN2017-82626-R.

\bibliographystyle{acm}
\bibliography{main}

\begin{thebibliography}{10}

\bibitem{arza_kernels_2020}
{\sc Arza, E., P{\'e}rez, A., Irurozki, E., and Ceberio, J.}
\newblock Kernels of {{Mallows Models}} under the {{Hamming Distance}} for
  solving the {{Quadratic Assignment Problem}}.
\newblock {\em Swarm and Evolutionary Computation\/} (July 2020), 100740.

\bibitem{benlic_memetic_2015}
{\sc Benlic, U., and Hao, J.-K.}
\newblock Memetic search for the quadratic assignment problem.
\newblock {\em Expert Systems with Applications 42}, 1 (Jan. 2015), 584--595.

\bibitem{beyer1991standard}
{\sc Beyer, W.~H.}
\newblock {\em Standard Probability and Statistics: {{Tables}} and Formulae}.
\newblock {CRC Press}, 1991.

\bibitem{brownLanguageModelsAre2020}
{\sc Brown, T.~B., Mann, B., Ryder, N., Subbiah, M., Kaplan, J., Dhariwal, P.,
  Neelakantan, A., Shyam, P., Sastry, G., Askell, A., Agarwal, S.,
  {Herbert-Voss}, A., Krueger, G., Henighan, T., Child, R., Ramesh, A.,
  Ziegler, D.~M., Wu, J., Winter, C., Hesse, C., Chen, M., Sigler, E., Litwin,
  M., Gray, S., Chess, B., Clark, J., Berner, C., McCandlish, S., Radford, A.,
  Sutskever, I., and Amodei, D.}
\newblock Language {{Models}} are {{Few-Shot Learners}}.

\bibitem{ceberioReviewDistancesMallows2015}
{\sc Ceberio, J., Irurozki, E., Mendiburu, A., and Lozano, J.~A.}
\newblock A review of distances for the {{Mallows}} and {{Generalized Mallows}}
  estimation of distribution algorithms.
\newblock {\em Computational Optimization and Applications 62}, 2 (Nov. 2015),
  545--564.

\bibitem{ceberio_linear_2015}
{\sc Ceberio, J., Mendiburu, A., and Lozano, J.~A.}
\newblock The linear ordering problem revisited.
\newblock {\em European Journal of Operational Research 241}, 3 (Mar. 2015),
  686--696.

\bibitem{conover1980practical}
{\sc Conover, W.~J.}
\newblock Practical nonparametric statistics.

\bibitem{dominguez_methodology_2012}
{\sc Dom{\'i}nguez, J., and Alba, E.}
\newblock A {{Methodology}} for {{Comparing}} the {{Execution Time}} of
  {{Metaheuristics Running}} on {{Different Hardware}}.
\newblock In {\em Evolutionary {{Computation}} in {{Combinatorial
  Optimization}}}, vol.~7245. {Springer Berlin Heidelberg}, {Berlin,
  Heidelberg}, 2012, pp.~1--12.

\bibitem{geeksknapsack}
{\sc {GeeksforGeeks}}.
\newblock 0-1 {{Knapsack Problem}}.
\newblock https://www.geeksforgeeks.org/0-1-knapsack-problem-dp-10/, 2022.

\bibitem{geeksqueens}
{\sc GeeksforGeeks}.
\newblock N {{Queen Problem}}.
\newblock https://www.geeksforgeeks.org/n-queen-problem-backtracking-3/, 2022.

\bibitem{Goldberg:1985:ATS:645511.657095}
{\sc Goldberg, D.~E., and Lingle, Jr., R.}
\newblock Alleles, {{Loci}} and the {{Traveling Salesman Problem}}.
\newblock In {\em Proceedings of the 1st {{International Conference}} on
  {{Genetic Algorithms}}\/} (1985), {L. Erlbaum Associates Inc.}, pp.~154--159.

\bibitem{gupta_flowshop_2006}
{\sc Gupta, J.~N., and Stafford, E.~F.}
\newblock Flowshop scheduling research after five decades.
\newblock {\em European Journal of Operational Research 169}, 3 (Mar. 2006),
  699--711.

\bibitem{doi:10.1126/science.359.6377.725}
{\sc Hutson, M.}
\newblock Artificial intelligence faces reproducibility crisis.
\newblock {\em Science (New York, N.Y.) 359}, 6377 (2018), 725--726.

\bibitem{koopmansAssignmentProblemsLocation1957}
{\sc Koopmans, T.~C., and Beckmann, M.}
\newblock Assignment {{Problems}} and the {{Location}} of {{Economic
  Activities}}.
\newblock {\em Econometrica 25}, 1 (Jan. 1957), 53.

\bibitem{nordhaus_2007}
{\sc Nordhaus, W.~D.}
\newblock Two centuries of productivity growth in computing.
\newblock {\em The Journal of Economic History 67}, 1 (2007), 128--159.

\bibitem{schiavinottoReviewMetricsPermutations2007}
{\sc Schiavinotto, T., and St{\"u}tzle, T.}
\newblock A review of metrics on permutations for search landscape analysis.
\newblock {\em Computers \& operations research 34}, 10 (2007), 3143--3153.

\bibitem{sharir2020cost}
{\sc Sharir, O., Peleg, B., and Shoham, Y.}
\newblock The cost of training nlp models: {{A}} concise overview.
\newblock {\em arXiv preprint arXiv:2004.08900\/} (2020).

\bibitem{wasserstein_asa_2016}
{\sc Wasserstein, R.~L., and Lazar, N.~A.}
\newblock The {{ASA Statement}} on p-{{Values}}: {{Context}}, {{Process}}, and
  {{Purpose}}.
\newblock {\em The American Statistician 70}, 2 (Apr. 2016), 129--133.

\bibitem{weicker_dhrystone_1988}
{\sc Weicker, R.~P.}
\newblock Dhrystone benchmark: Rationale for version 2 and measurement rules.
\newblock {\em ACM SIGPLAN Notices 23}, 8 (Aug. 1988), 49--62.

\bibitem{Charles2013}
{\sc Zeggel, M.}
\newblock Magic-square.
\newblock https://github.com/Gizmoscope/magic-square, 2019.

\bibitem{zou2003correlation}
{\sc Zou, K.~H., Tuncali, K., and Silverman, S.~G.}
\newblock Correlation and simple linear regression.
\newblock {\em Radiology 227}, 3 (2003), 617--628.

\end{thebibliography}

\newpage

\appendix
\FloatBarrier

\section{The importance of using the same resources in algorithm comparison}
\label{appendix:extra_execution_time_typeI_error}
\begin{figure}
	\centering
	\caption*{\hspace{1.6em}Increase in the probability type I error with respect \\  \hspace{1.6em} to the difference in runtime}
	\vspace{-1em}
	\includegraphics[width=0.7\linewidth]{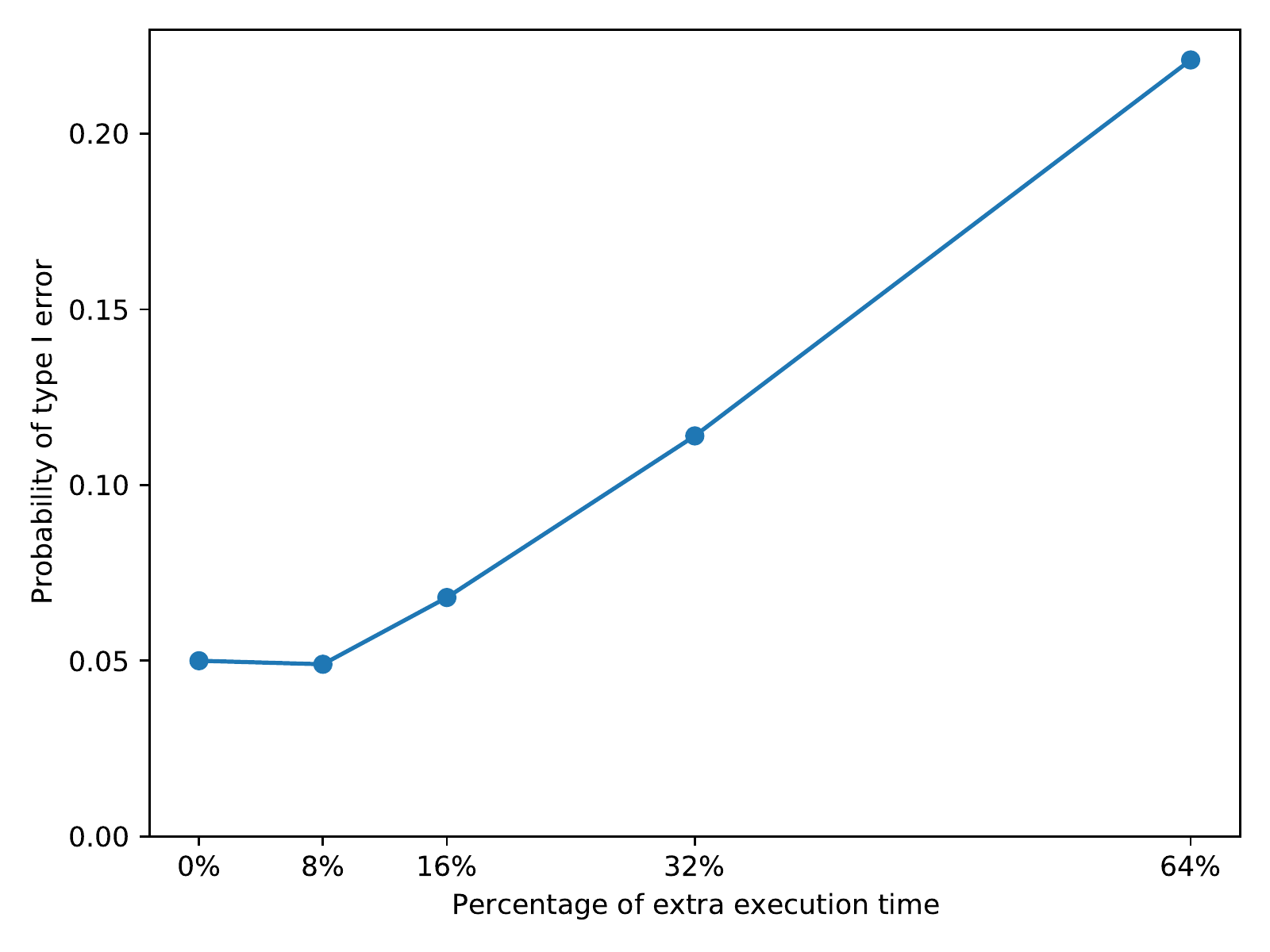}
	\caption{
		Probability of type I error in the one-sided sign test when comparing two identical random search algorithms.
		One of the algorithms is given extra runtime, according to the $x$-axis.		
		The test is applied to a set of 16 problem instances.
	}
	\label{fig:extra_execution_time_typeI_error}
\end{figure}

It is essential to run the algorithms with the same computational resources to carry out a fair comparison.
To better illustrate this point, in the following lines, a small experiment is presented.
This experiment illustrates the increase in the probability of type I error (the probability of erroneously concluding a difference in performance, when in reality, there is none) with respect to the difference in execution time.
Specifically, we run a random search algorithm twice in each problem instance\footnote{A set of 16 permutation problem instances is considered, 4 instances of 4 problems. The four permutation problems considered are the traveling salesman problem, the permutation flowshop scheduling problem, the linear ordering problem, and the quadratic assignment problem.} and perform the one-sided sign test~\cite{conover1980practical}\footnote{In \ref{appendix:why_sign_test}, we explain why we limit the statistical analysis to the sign test in this paper.} (see Section~\ref{section:corrected_critical_value} for an explanation of the sign test), on the set of results obtained.	
The significance level is set to $\alpha = 0.05$.
Even though the random search algorithm is being compared with itself, we increase the runtime of one of the executions by 8, 16, 32, or 64 percent.
We repeat the steps above $1000$ times to estimate the probability of type I error (estimated as the probability of rejecting $H_0$).

Figure~\ref{fig:extra_execution_time_typeI_error} shows the estimated probability of type I error.
Notice that the type I error starts at $0.05$, which is the expected result for a significance level of $\alpha = 0.05$.
However, the error shoots up dramatically when the difference in runtime increases, more than doubling when the percentage of extra runtime reaches 32\%.
Therefore, a discrepancy in the runtime of the algorithms being compared, if high enough, can lead to falsely concluding that the performance of the algorithms is not the same.
A fair comparison requires the same computational resources to be assigned in the execution of each algorithm.

\section{Justification of Assumption 1}
\label{appendix:constant_ratio_tasks}

	The runtime of an optimization process (a sequence of computational instructions) is different in each machine.
However, even though it is different, there might be a proportional relationship between the runtime of the same optimization process in two different machines.
This hypothesis is the basis of Assumption~\ref{assumption:k_constant_ratio}.

To experimentally study this assumption, we compute the correlation of the runtime that several optimization processes have on two machines. 
Specifically, we computed the correlation of 64 different optimization processes (see Appendix~\ref{appendix:experimentation_execution} for additional details on the optimization processes) for every possible pair of machines from the \nmachines\ different machines used in the experimentation.
The average Pearson's correlation coefficient of the runtimes is 0.989987, which shows a strong linear~\cite{zou2003correlation} (not necessarily proportional) relationship between the runtime of the same optimization process in two different machines.

\begin{figure}
	\centering
	\caption*{$\hspace{1.9em}$Estimating the equivalent runtime}
	\includegraphics[width=0.6\linewidth]{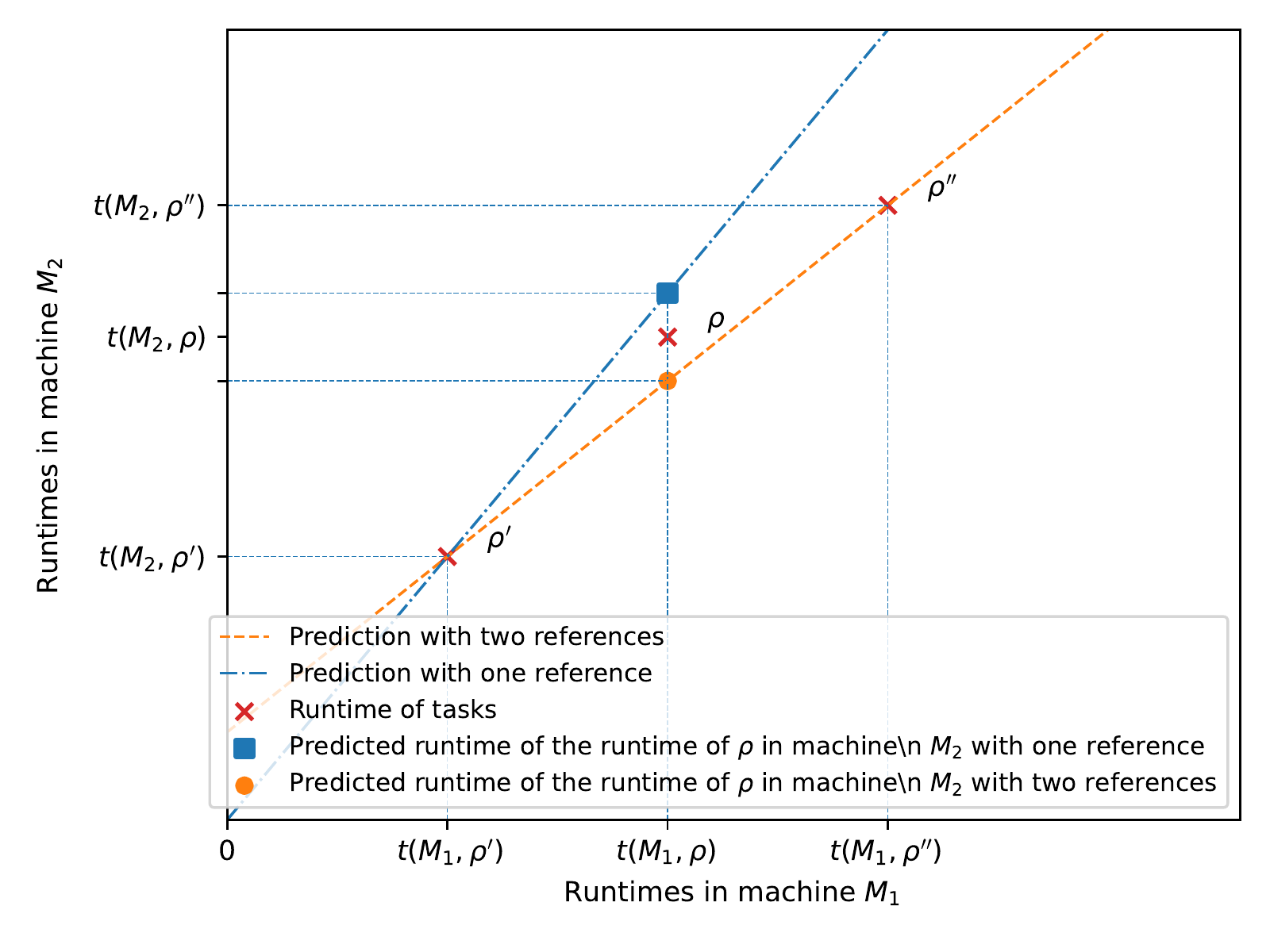}
	\caption{
		Estimation of the runtime of an optimization process $\rho$ with one $\rho'$ or two $\rho',\rho''$ reference optimization processes.
		The $x$-axis represents the runtime in machine $M_1$, while the y-axis is the runtime in machine $M_2$.
		The runtime of the optimization process $\rho$ is estimated for machine $M_2$.	
	}
	\label{fig:estimation_explained}
\end{figure}

Given two machines $M_1$ and $M_2$, the runtime of an optimization process can be considered as a two-dimensional vector, where each of the dimensions represents the runtime of the optimization process in each of the machines.
Thus, knowing the runtime $t(s,M_1)$ of an optimization process $\rho$ in a machine $M_1$, it is reasonable to estimate the equivalent runtime of $\rho$ in another machine $M_2$, when the runtime of two other optimization processes $\rho'$ and $\rho''$ is known for both machines.
In fact, with such a high Pearson's correlation coefficient, the runtimes of these optimization processes (red crosses in Figure~\ref{fig:estimation_explained}) will almost be aligned in a line~\cite{zou2003correlation}.
Therefore, the estimated runtime of $\rho$ in machine $M_2$ is defined as the value that makes the runtime of the three optimization processes aligned.
This is shown by the orange line in Figure~\ref{fig:estimation_explained}.

Observe that this procedure requires the runtime of two optimization processes $\rho'$ and $\rho''$ to be known in both machines $M_1,M_2$.
However, by considering an additional hypothesis, we can reduce the requirement to only one optimization process $\rho'$.
This additional hypothesis is that the regression line has to cross the origin.
Intuitively, if an optimization process (sequence of computational instructions) takes no time in a machine, it makes no sense that it takes a positive amount of time in another machine.
In addition, without this condition, it could be possible to estimate a negative runtime, which is not properly defined.
%
%

In this setting, the estimated runtime for the optimization process $\rho$ in machine $M_2$ is set so that the runtime of the optimization processes $\rho$ and $\rho'$ and the origin are in the same line.
This is represented by the blue line in Figure~\ref{fig:estimation_explained}. 
The estimation of the runtime of the optimization process $\rho$ in machine $M_2$, shown in the figure as a blue square, is given by the slope-intercept formula for the points $(0,0)$ and $(t(M_1,\rho'), t(M_2,\rho'))$:


\begin{equation}
\label{equation:k_constant_ratio_appendix}
t(M_2,\rho) \approx \frac{t(M_2,\rho')}{t(M_1,\rho')}t(M_1,\rho) 
\end{equation}

By rewriting Equation~\eqref{equation:k_constant_ratio_appendix}, we obtain that the ratio of two optimization processes is (approximately) constant in different machines

\begin{equation}
\label{equation:k_constant_ratio_appendix_mod}
\frac{t(M_2,\rho)}{t(M_2,\rho')} \approx \frac{t(M_1,\rho)}{t(M_1,\rho')} 
\end{equation}

which is exactly Assumption~\ref{assumption:k_constant_ratio}.

\section{Optimization processes}
\label{appendix:experimentation_execution}

As defined in this paper, an optimization process is just a sequence of computational instructions that can be executed in any machine.
Specifically, each of the optimization processes described in this section consists of executing an optimization algorithm in a problem instance for a maximum of $2\cdot10^6$ objective function evaluations.
In total, we considered 64 optimization processes, executing 4 algorithms in 16 problem instances.
The optimization process $\rho'$ is the sequential execution of these 64 optimization processes.

\textbf{Problem instances}:
We solved four types of optimization problems (all of them are permutation problems): the traveling salesman problem~\cite{Goldberg:1985:ATS:645511.657095}, the quadratic assignment problem~\cite{koopmansAssignmentProblemsLocation1957}, the linear ordering problem~\cite{ceberio_linear_2015} and the permutation flowshop scheduling problem~\cite{gupta_flowshop_2006}.
For each of these four problem types, we chose 4 problem instances, as listed in Table~\ref{table:instance_list}.

\textbf{Optimization algorithms}:
Each of the 16 problem instances was optimized with four optimization algorithms.
These optimization algorithms are random search and local search with three different neighborhoods: swap, interchange, and insert~\cite{schiavinottoReviewMetricsPermutations2007,ceberioReviewDistancesMallows2015}.
The local search is a best-first or greedy approach that is randomly reinitialized when a local optimum is found.

We define each of the 64 different optimization processes as running each of these four optimization algorithms in each of the 16 problem instances.

	\begin{table}[]
	\centering
	\begin{footnotesize}
		\caption*{\small \textbf{Problem instances}}
		\begin{tabular}{|l|c|c|}
			\hline
			\textbf{instance name}   & \textbf{problem} &     \textbf{size} \\ \hline
			tai75e02        &   qap   &       $\ $ 75 \\
			sko100a         &   qap   &      100 \\
			tai100a         &   qap   &      100 \\
			tai100b         &   qap   &      100 \\\hline
			eil101          &   tsp   &      101 \\
			pr136           &   tsp   &      136 \\
			kroA200         &   tsp   &      200 \\
			kroB200         &   tsp   &      200 \\\hline
			tai100\_20\_0   &  pfsp   & (100,20) \\
			tai100\_20\_1   &  pfsp   & (100,20) \\
			tai200\_20\_1   &  pfsp   & (200,20) \\
			tai200\_20\_1   &  pfsp   & (200,20) \\\hline
			N-be75np\_150   &   lop   &      150 \\
			N-stabu3\_150   &   lop   &      150 \\
			N-t65d11xx\_150 &   lop   &      150 \\
			N-t70f11xx\_150 &   lop   &      150 \\ \hline
		\end{tabular}
	\end{footnotesize}
	\caption
	{
		The list of 16 problem instances and their size.
	}
	\label{table:instance_list}
	
\end{table}

\textbf{Machines:}
The experimentation was carried out in a set of \nmachines\ different machines.
Table~\ref{table:machine_list} lists the CPU models of these machines, as well as their single-thread PassMark CPU scores.

	\begin{table}[]
	\centering
	\begin{footnotesize}
		\caption*{\small \textbf{Machines}}
		\begin{tabular}{lr}
			\textbf{CPU model name}    & \textbf{PassMark score} \\
			Intel i5 470U              &                  539 \\
			Intel Celeron N4100        &                 1012 \\
			AMD A9 9420 with Radeon R5 &                 1344 \\
			AMD FX 6300                &                 1486 \\
			Intel i7 2760QM            &                 1559 \\
			Intel i7 6700HQ (2.60GHz)  &                 1921 \\
			Intel i7 7500U             &                 1955 \\
			AMD Ryzen7 1800X           &                 2185 
		\end{tabular}
	\end{footnotesize}
	\caption
	{
		The list of \nmachines\ machines used in the experimentation and their speed score, measured in terms of PassMark single-thread score.
	}
	\label{table:machine_list}
	
\end{table}

\section{The sign test for algorithm performance comparison}
\label{appendix:why_sign_test}

When statistically assessing the comparison of the performance of optimization algorithms, a classical way is to use non-parametric tests as the distribution of the performance is usually unknown.
In the literature, the Wilcoxon signed-rank test, the Mann-Whitney test and the sign test~\cite{conover1980practical} are often used to assess a statistically significant difference in the performance of two algorithms.
We argue that, in the context of optimization algorithm performance comparison, it may be more suitable to use the sign test than the Wilcoxon signed-rank or the Mann-Whitney test.

It turns out that the result of the Wilcoxon and the Mann-Whitney tests might change when the objective function value of some of the problems is scaled (multiplied or divided by a positive constant).
The reason is that they both take into account the magnitude of the differences between the observations, and these differences change with scaling.
A usual solution is to consider the average relative deviation percentage with respect to the optimum (or any other reference solution) instead of the objective value, but this only changes the problem: now the results of these tests change when the objective function value of some of the problems is shifted (add or subtract a constant).
In our opinion, the performance comparison of two optimization algorithms should be invariant to these two alterations, otherwise, problems that are on a higher scale (for example, when the dimension of the problem is high), will have a larger impact on the result of the statistical test.
In addition, we believe that it is reasonable that all problem instances have the same weight in the conclusion of the statistical test, which both the  Wilcoxon signed-rank and the Mann-Whitney test are unable to accomplish due to their dependence on the magnitude of the differences.

An alternative is the sign test~\cite{conover1980practical}, which is invariant to the shifting and scaling of the problems.
In fact, the result of the sign test does not change even if some of the problems are modified by composing the objective function with any strictly increasing function.
For this reason, and even though the sign test is a less powerful alternative (higher probability of type II error), we believe it is the most suitable hypothesis test for algorithm performance comparison when the objective functions of all the problems are not directly comparable.

\pagebreak
\section{Proof of Equation~\eqref{equation:p_value_bound_sign_test}.}
\label{appendix:proof_inequality_p_corrected}
When performing the statistical analysis, a set of $n$ problem instances is used to compute the statistic and the \textit{p-}value.
The goal of the analysis is to draw conclusions on a larger set of problem instances based on the observed sample of size $n$.
%
%
%
%
%
Given a problem instance, we can define the performance of an algorithm in this instance.

\begin{mydef}
	\label{definition:performance_algorithm}
	(The performance of an algorithm in an instance) \\
	Let $M$ be a machine, $t$ a stopping criterion in terms of maximum runtime, $A$ an optimization algorithm and $i$ a problem instance. 
	The performance of algorithm $A$ in an instance $i$, denoted $A(M,t,i)$, is defined as a random variable whose outcome is obtained by first sampling a random seed $r$ and then optimizing instance $i$ with optimization algorithm $A$ in machine $M$ for time $t$.
	Given this random seed $r$, the performance of an algorithm in an instance is deterministic.
\end{mydef}

In Section~\ref{section:norm_exec_time}, we defined $t_1$ as the stopping criterion for algorithm $A$ in machine $M_1$, which is obviously the time it takes to carry out this optimization process in machine $M_1$.
We also defined the equivalent runtime $t_2$ as the time it takes to replicate the exact same optimization process in machine $M_2$ in Definition~\ref{definition:equivalent_runtime}.
Because of this definition, $A(M_1,t_1,i)$ and $A(M_2,t_2,i)$ are the same random variables.
Therefore, it makes sense to denote $A(M_1,t_1,i)$ and $A(M_2,t_2,i)$ or $B(M_1,t_1,i)$ and $B(M_2,t_2,i)$ as $A_i$ or $B_i$, respectively.
To ease the notation, we will also denote $B(M_2,\hat{t}_2,i)$ as $\hat{B}_i$.

Finally, as discussed in Section~\ref{section:limitations_sign_test}, we assume that whether $\hat{t}_2 < t_2$ is true or not is independent for each instance $i$, and that $\mathcal{P}(\hat{t}_2 < t_2) < p_{\gamma}$.
Let us now prove Equation~\eqref{equation:p_value_bound_sign_test}.


\begin{mylemma}
	\label{lemma:implication_implies_lower_prob_lower_count}
	Let $n$ be an integer, $X$ and $Y$ two random variables.
	Let $X_1,...,X_n$ be $n$ independent random variables distributed as $X$. 
	Let $Y_1,...,Y_n$ be $n$ independent random variables distributed as $Y$.
	Let $v_x$ and $v_y$ be two possible outcomes of the random variables $X$ and $Y$ respectively, $l \in \{0,...,n\}$ be an integer and $p \in (0,1)$ be a real number. 
	
	I) If $\mathcal{P}[Y = v_y \ | \ X = v_x] = 1$, then
	
	$$\mathcal{P}[X = v_x] \leq \mathcal{P}[Y = v_y]$$
	and
	$$\#\{X_i = v_x\} \leq \#\{Y_i = v_y\}$$
	\vspace{0.25cm}
	
	II) If $\mathcal{P}[Y = v_y \ | \ X = v_x] = 1$ and $\mathcal{P}[X = v_x \ | \ Y = v_y] = 1$ then
	
		$$\mathcal{P}[X = v_x] = \mathcal{P}[Y = v_y]$$\vspace{0.25cm}
	
%
	
	III) If $\mathcal{P}[X = v_x] < p$ then
	
	$$\mathcal{P}[\#\{X_i = v_x\} \geq l ] < \mathcal{P}[Bin(n,p) \geq l ]$$

\end{mylemma}

\begin{mylemma}
	\label{lemma:MinKV}
	Let $i \in \{1,..,n\}$ be $n$ problem instances and let $A$ and $B$ be two optimization algorithms.
	Let $a_i$, $b_i$ and $\hat{b}_i$ be the observed values of $A_i$, $B_i$ and $\hat{B}_i$ respectively, $\forall i \in \{1,...,n\}$.
	Let $k$ and $v \in \{0,...,n\}$ be two integers.
	Suppose that $A_i \neq B_i$ and $A_i \neq \hat{B}_i$.

	Then,
		
	$$\mathcal{P}[\#\{A_i<min(\hat{B}_i,B_i)\} \leq \min(k,v) \lwhere \#\{A_i < B_i\}=v] \leq $$
	$$\mathcal{P}[\#\{A_i>\hat{B}_i \land A_i < B_i\}  \geq \max(0,v-k) \lwhere \#\{A_i < B_i\}=v] $$

\end{mylemma}	
\begin{proof}

	$$\#\{A_i<min(\hat{B}_i,B_i)\} \leq \min(k,v) \implies$$
	$$ \#\{A_i<\hat{B}_i \land A_i < B_i\} \leq \min(k,v) \implies$$
	$$ \#\{ A_i < B_i\} - \#\{A_i > \hat{B}_i \land A_i < B_i\} \leq \min(k,v) \implies$$

	Substituting $\#\{A_i < B_i\}=v$, 
	
	$$ v -  \min(k,v) \leq   \#\{A_i > \hat{B}_i \land A_i < B_i\} \implies$$
	
	Considering $v -  \min(k,v) = \max(0,v-k)$, 

	$$\#\{A_i > \hat{B}_i \land A_i < B_i\} \geq  \max(0,v-k)$$

	We have just shown that 
	
	$$\#\{A_i<min(\hat{B}_i,B_i)\} \leq \min(k,v) \implies \#\{A_i > \hat{B}_i \land A_i < B_i\} \geq  \max(0,v-k)$$
	
	Which means that, 
	
	$$\mathcal{P}[\#\{A_i > \hat{B}_i \land A_i < B_i\} \geq  \max(0,v-k) \lwhere \#\{A_i<min(\hat{B}_i,B_i)\} \leq \min(k,v)] = 1$$

	Finally, we apply Lemma~\ref{lemma:implication_implies_lower_prob_lower_count} I), obtaining

	$$\mathcal{P}[\#\{A_i<min(\hat{B}_i,B_i)\} \leq \min(k,v) \lwhere \#\{A_i < B_i\}=v] \leq $$
	$$\mathcal{P}[\#\{A_i>\hat{B}_i \land A_i <B_i\}  \geq \max(0,v-k) \lwhere \#\{A_i < B_i\}=v] $$

\end{proof}

\begin{mylemma}
	\label{lemma:bound_coef_of_theorem}
	Let $i \in \{1,..,n\}$ be $n$ problem instances and let $A$ and $B$ be two optimization algorithms.
	Let $a_i$, $b_i$ and $\hat{b}_i$ be the observed values of $A_i$, $B_i$ and $\hat{B}_i$ respectively, $\forall i \in \{1,...,n\}$.
	Let $k$ and $v \in \{0,...,n\}$ be two integers.
	Suppose that $A_i \neq B_i$ and $A_i \neq \hat{B}_i$.

	Then,
	
	$$\mathcal{P}[\#\{A_i<\hat{B}_i\} \leq k \lwhere \#\{A_i < B_i\}=v] < 	\mathcal{P}[Bin(n,p_{\gamma}) \geq \max(0,v-k)]$$
\end{mylemma}	
\begin{proof}
	$$\mathcal{P}[\#\{A_i<\hat{B}_i\} \leq k \lwhere \#\{A_i < B_i\}=v] \leq$$
	$$\mathcal{P}[\#\{A_i<min(\hat{B}_i,B_i)\} \leq k \lwhere \#\{A_i < B_i\}=v]$$

	Now, observe that $\#\{A_i < \min(B_i,\hat{B}_i)\} \leq \#\{A_i < B_i\} = v$, which implies that 
	
	$$\#\{A_i<min(\hat{B}_i,B_i)\} \leq k \iff \#\{A_i<min(\hat{B}_i,B_i)\} \leq min(k,v)$$

	This means that 
	
	$$ \mathcal{P}[\#\{A_i<min(\hat{B}_i,B_i)\} \leq k \lwhere \#\{A_i<min(\hat{B}_i,B_i)\} \leq min(k,v) \land \#\{A_i < B_i\}=v] = 1$$ and
	$$ \mathcal{P}[\#\{A_i<min(\hat{B}_i,B_i)\} \leq min(k,v) \lwhere \#\{A_i<min(\hat{B}_i,B_i)\} \leq k  \land \#\{A_i < B_i\}=v] = 1$$
	
	We apply Lemma~\ref{lemma:implication_implies_lower_prob_lower_count} II), obtaining
	
	$$\mathcal{P}[\#\{A_i<min(\hat{B}_i,B_i)\} \leq k \lwhere \#\{A_i < B_i\}=v] = $$
	$$\mathcal{P}[\#\{A_i<min(\hat{B}_i,B_i)\} \leq \min(k,v) \lwhere \#\{A_i < B_i\}=v]$$

	Applying Lemma~\ref{lemma:MinKV}, we obtain

	$$\mathcal{P}[\#\{A_i<min(\hat{B}_i,B_i)\} \leq \min(k,v) \lwhere \#\{A_i < B_i\}=v] \leq $$
	$$\mathcal{P}[\#\{A_i>\hat{B}_i \land A_i < B_i\}  \geq \max(0,v-k) \lwhere \#\{A_i < B_i\}=v] $$

	Note that $b_i$ is the score obtained with the true equivalent runtime $t_2$ as the stopping criterion, while in the case of $\hat{b}_i$, the stopping criterion is the estimated equivalent runtime $\hat{t}_2$.
	In a minimization context,  $\hat{b}_i < b_i \implies \hat{t_2} > t_2$, because a better score can only be obtained with a longer runtime (a shorter runtime implies an equal or worse performance).
	Let us consider the following implications: 
	$$a_i>\hat{b}_i \land a_i < b_i \implies \hat{b}_i < b_i \implies \hat{t_2} > t_2$$
	We infer that
	$$\mathcal{P}[\hat{t_2} > t_2 \lwhere A_i>\hat{B}_i \land A_i < B_i] = 1$$
	Applying Lemma~\ref{lemma:implication_implies_lower_prob_lower_count} I), we obtain  
	$$\mathcal{P}[\#\{A_i>\hat{B}_i \land A_i < B_i \lwhere \#\{A_i < B_i\}=v \} \geq \max(0,v-k)] \leq $$
	$$\mathcal{P}[\#\{\hat{t_2} > t_2 \lwhere \#\{A_i < B_i\}=v \} \geq \max(0,v-k)] = $$
	$$\mathcal{P}[\#\{\hat{t_2} > t_2\} \geq \max(0,v-k)]$$
	
	The estimated runtime $\hat{t}_2$ was computed with the equation in Definition~\ref{definition:estimation_execution_time} in Section~\ref{section:norm_exec_time}, with an estimated probability that $\hat{t}_2 < t_2$ lower than $0.01$.
	With this information, we apply Lemma~\ref{lemma:implication_implies_lower_prob_lower_count} III) taking into account that $\mathcal{P}[\hat{t}_2 > t_2] < 0.01$:

	$$\mathcal{P}[\#\{\hat{t_2} > t_2\} \geq \max(0,v-k)] < $$
	$$\mathcal{P}[Bin(n,0.01) \geq \max(0,v-k)]$$

\end{proof}

\begin{mytheorem}
	\label{theorem:corrected_p_value_upper_bound}
	Let $i \in \{1,..,n\}$ be $n$ problem instances and let $A$ and $B$ be two optimization algorithms.
	Let $a_i$, $b_i$ and $\hat{b}_i$ be the observed values of $A_i$, $B_i$ and $\hat{B}_i$ respectively, $\forall i \in \{1,...,n\}$.
	Let $H_0$ be the null hypothesis under which the statistic $\#\{A_i<B_i\}$ follows the null distribution $Bin(n,0.5)$.
	Suppose that $A_i \neq B_i$ and $A_i \neq \hat{B}_i$.
	Then,
	
	$$\mathcal{P}[\#\{A_i<\hat{B}_i\} \leq k \lwhere H_0] \leq $$
	$$\sum_{v=0}^{n} (1 - \mathcal{P}[Bin(n,0.01) < \max(0,v-k)]) \cdot \mathcal{P}[Bin(n,0.5) = v]$$

\end{mytheorem}
\begin{proof}
	Let $X,C$ be a two random variables, where $S_C$ and $S_X$ are the sets of all possible outcomes of $C$ and $X$ respectively.
	Consider the law of total probability~\cite{beyer1991standard}: 
	
	$$\forall x \in S_X \ , \  \mathcal{P}[X = x] = \sum_{c \in S_C}  \mathcal{P}[C = c] \cdot \mathcal{P}[X = x \lwhere C = c]$$
	
	Applying this formula, we obtain
	
	$$\mathcal{P}[\#\{A_i<\hat{B}_i\} \leq k \lwhere H_0] = $$
	$$\sum_{v=0}^{n} \mathcal{P}[\#\{A_i<\hat{B}_i\} \leq k \lwhere H_0 \land \#\{A_i < B_i\}=v] \cdot \mathcal{P}[\#\{A_i < B_i\} = v \lwhere H_0]$$
	
	Given that $\#\{A_i < B_i\} = v$, we can say that $\#\{A_i<\hat{B}_i\} \leq k$ is independent of $H_0$, because $\#\{A_i<\hat{B}_i\}$ is determined by how many times $\hat{t}_2 > t_2$ resulted in $A_i < B_i \land A_i > \hat{B}_i$ and  $\hat{t}_2 < t_2$ resulted in $A_i > B_i \land A_i < \hat{B}_i$. 
	Specifically, $H_0$ gives the prior probabilities of $A_i > B_i$, which are not relevant when we know that $\#\{A_i > B_i\} = v$.
	That gives us	
	
	$$\sum_{v=0}^{n} \mathcal{P}[\#\{A_i<\hat{B}_i\} \leq k \lwhere H_0 \land \#\{A_i < B_i\}=v] \cdot \mathcal{P}[\#\{A_i < B_i\} = v \lwhere H_0] = $$
	$$\sum_{v=0}^{n} \mathcal{P}[\#\{A_i<\hat{B}_i\} \leq k \lwhere \#\{A_i < B_i\}=v] \cdot \mathcal{P}[\#\{A_i < B_i\} = v \lwhere H_0]$$
	
	Applying Lemma~\ref{lemma:bound_coef_of_theorem} and considering that $H_0$ implies the null distribution $Bin(n,0.5)$ for the statistic $\#\{A_i<B_i\}$,
	
	$$\sum_{v=0}^{n} \mathcal{P}[\#\{A_i<\hat{B}_i\} \leq k \lwhere \#\{A_i < B_i\}=v] \cdot \mathcal{P}[\#\{A_i < B_i\} = v \lwhere H_0] < $$
	$$\sum_{v=0}^{n} \mathcal{P}[Bin(n,0.01) \geq \max(0,v-k)] \cdot \mathcal{P}[\#\{A_i < B_i\} = v \lwhere H_0] = $$
	$$\sum_{v=0}^{n} \mathcal{P}[Bin(n,0.01) \geq \max(0,v-k)] \cdot \mathcal{P}[Bin(n,0.5) = v] = $$
	$$\sum_{v=0}^{n} (1 - \mathcal{P}[Bin(n,0.01) < \max(0,v-k)]) \cdot \mathcal{P}[Bin(n,0.5) = v]$$

\end{proof}

\end{document}